\documentclass[preprint]{elsarticle}

\usepackage{amssymb}
\usepackage{amsfonts}
\usepackage{amsmath}
\usepackage{amsthm}
\usepackage{bm}
\usepackage{color}      
\usepackage{enumerate}
\usepackage{graphicx}
\usepackage{stackengine}
\usepackage{subfig}
\usepackage{algorithm}
\makeatletter
\def\ps@pprintTitle{%
  \let\@oddhead\@empty
  \let\@evenhead\@empty
  \let\@oddfoot\@empty
  \let\@evenfoot\@oddfoot
}
\makeatother

\journal{Journal of Computational Physics}

\newtheorem{thm}{Theorem}
\newtheorem{lem}[thm]{Lemma}
\newtheorem{prop}[thm]{Proposition}
\newtheorem{cor}[thm]{Corollary}

\newtheorem{remark}[thm]{Remark}

\newcommand{\tr}{\mathrm{tr}}        
\newcommand{\diag}{\mathrm{diag}}        
\newcommand{\dd}{:} 
\newcommand{\tp}{^{T}} 

\DeclareMathOperator*{\argmin}{{\arg \min}}
\newcommand{\trid}{
	\mathrel{\stackunder[-2pt]{\stackon[-2pt]{$\cdot$}{$\cdot$}}{$\cdot$}}}
\newcommand{\va}{\mathbf{a}}
\newcommand{\vb}{\mathbf{b}}

\newcommand{\vp}{\mathbf{p}}

\newcommand{\vq}{\mathbf{q}}
\newcommand{\ve}{\mathbf{e}}
\newcommand{\vQ}{\mathbf{Q}}
\newcommand{\vr}{\mathbf{r}}

\newcommand{\valpha}{\bm{\alpha}}

\newcommand{\vm}{\mathbf{m}}
\newcommand{\vn}{\mathbf{n}}

\newcommand{\vx}{\mathbf{x}}

\newcommand{\vLambda}{\mathbf{\Lambda}}
\newcommand{\vXi}{\mathbf{\Xi}}
\newcommand{\vPi}{\mathbf{\Pi}}

\newcommand{\vA}{\mathbf{A}}

\newcommand{\vC}{\mathbf{C}}
\newcommand{\vD}{\mathbf{D}}
\newcommand{\vE}{\mathbf{E}}

\newcommand{\vP}{\mathbf{P}}

\newcommand{\vI}{\mathbf{I}}
\newcommand{\vR}{\mathbf{R}}
\newcommand{\vT}{\mathbf{T}}

\newcommand{\vH}{\mathbf{H}}
\newcommand{\vzero}{\mathbf{0}}

\newcommand{\Om}{\Omega}

\newcommand{\Gm}{\Gamma}

\newcommand{\iO}{\int_{\Omega}}

\newcommand{\iG}{\int_{\Gamma}}

\newcommand{\Sp}{S^{2}}  
\newcommand{\iSp}{\int_{\Sp}}
\newcommand{\dA}[1]{\, dS(#1)}

\newcommand{\Prob}{\mathcal{P}}   
\newcommand{\M}{\mathbb{M}}   
\newcommand{\Sh}{\mathbb{S}_h}   
\newcommand{\V}{\mathbb{V}}   
\newcommand{\R}{\mathbb{R}}   

\newcommand{\dt}{\delta t}

\newcommand{\ipOm}[2]{\left( #1, #2 \right)}        

\newcommand{\ebar}{\bar{\varepsilon}}
\newcommand{\ea}{\varepsilon_{\mathrm{a}}}

\newcommand{\etens}{\bm{\varepsilon}}

\newcommand{\Tk}{\mathcal{T}}

\newcommand{\Pk}{\mathcal{P}}
\newcommand{\interp}{\mathcal{I}_{h}} 

\newcommand{\symmtraceless}{\vXi}
\newcommand{\basis}{\vE}
\newcommand{\ELdG}{E}

\newcommand{\ELdGfunc}{\mathcal{W}}
\newcommand{\aform}[2]{a \left( #1 , #2 \right)}
\newcommand{\bform}[2]{b \left( #1 , #2 \right)}
\newcommand{\cform}[3]{c \left( #1 , #2 ; #3 \right)}
\newcommand{\LdGspace}[1]{\V (#1)}
\newcommand{\LdGsurf}{f_{\Gm}}
\newcommand{\surfcoef}{\eta_{\Gm}}
\newcommand{\LdGrhs}{\chi}
\newcommand{\Li}{L}
\newcommand{\levi}{e}
\newcommand{\bflevi}{\ve}
\newcommand{\Gmdir}{\Gamma_{D}}
\newcommand{\vQdir}{\vQ_{D}}

\newcommand{\vQinit}{\vQ_{0}}

\newcommand{\vQgam}{\vQ_{\Gm}}

\newcommand{\AdmisLdG}{\mathcal{A}}   

\newcommand{\Bulkfunc}{\psi}
\newcommand{\fbulk}{f}
\newcommand{\bulkimp}{\psi_{c}}
\newcommand{\bulkexp}{\psi_{e}}
\newcommand{\bulkeps}{\epsilon}
\newcommand{\bulkA}{A}
\newcommand{\bulkB}{B}
\newcommand{\bulkC}{C}

\newcommand{\bulkK}{K}

\newcommand{\entropy}[1]{S[#1]}
\newcommand{\MS}[1]{I_{\mathrm{MS}}[#1]}
\newcommand{\partition}[1]{Z(#1)}
\newcommand{\Eval}[2]{\mathbb{E}_{#1} [#2]}

\newcommand{\Admisrho}[1]{\mathcal{A}_{#1}}   

\newcommand{\Lagr}[2]{L \left[ #1, #2 \right]}   
\newcommand{\Constr}[1]{\vC \left[ #1 \right]}   
\newcommand{\dual}[1]{W \left( #1 \right)}   

\graphicspath{{../Figures/}}

\begin{document}

\begin{frontmatter}


\title{Numerical method for the equilibrium configurations of a
Maier-Saupe bulk potential in a Q-tensor model of an anisotropic nematic 
liquid crystal}


\author{Cody D. Schimming\fnref{CSfootnote}}
\address{School of Physics and Astronomy \\
  University of Minnesota, 
  Minneapolis, MN 55455}
\ead{schim111@umn.edu}
\fntext[CSfootnote]{C. D. Schimming acknowledges financial support by NSF DMR-1838977.}

\author{Jorge Vi\~{n}als\fnref{JVfootnote}}
\address{School of Physics and Astronomy \\
  University of Minnesota,
  Minneapolis, MN 55455}
\fntext[JVfootnote]{J. Vi\~{n}als acknowledges financial support by NSF DMR-1838977.}

\author{Shawn W.~Walker\corref{corrauthor}}
\address{Department of Mathematics\\
Center for Computation and Technology (CCT)\\
Louisiana State University,\\
Baton Rouge, LA 70803
Tel.: +1-225-578-1603\\
Fax: +1-225-578-4276
}
\cortext[corrauthor]{Corresponding author; S. W. Walker acknowledges financial support by the NSF DMS-1555222 (CAREER).}
\ead{walker@math.lsu.edu}


%

\begin{abstract}

We present a numerical method, based on a tensor order parameter description of
a nematic phase, that allows fully anisotropic elasticity. Our method thus
extends the Landau-de Gennes $\vQ$-tensor theory to anisotropic phases. A
microscopic model of the nematogen is introduced (the Maier-Saupe potential in
the case discussed in this paper), combined with a constraint on eigenvalue
bounds of $\vQ$. This ensures a physically valid order parameter $\vQ$ (i.e.,
the eigenvalue bounds are maintained), while allowing for more general gradient
energy densities that can include cubic nonlinearities, and therefore elastic
anisotropy. We demonstrate the method in two specific two dimensional
examples in which the Landau-de Gennes model including elastic anisotropy is
known to fail, as well as in three dimensions for the cases of a hedgehog point
defect, a disclination line, and a disclination ring. The details of the 
numerical implementation are also discussed.

\end{abstract}


\end{frontmatter}

\section{Introduction}\label{sec:intro}


Liquid crystals (LCs) are a critical material for emerging technologies \cite{deGennes_book1995,Lagerwall_CAP2012}.  Their response to optical \cite{Blinov_book1983,Goodby_inbook2012,Sun_SMS2014,Hoogboom_RSA2007,Dasgupta_MRE2015}, electric/magnetic \cite{Brochard_JPhysC1975,Buka_book2012,Shah_Small2012}, and mechanical actuation \cite{Zhu_PRE2011,LC_Elastomers_book2012,Biggins_JMPS2012,Resetic_NC2016} has already yielded various devices, e.g. electronic shutters \cite{Heo_AA2015}, novel types of lasers \cite{Humar_OE2010,Coles_NP2010}, dynamic shape control of elastic bodies \cite{Camacho-Lopez_NM2004,Ware_Science2015}, and others \cite{Musevic2011,Lopez-Leon_CPS2011,Copar_PNAS2015,Whitmer_PRL2013,Wang_NL2014}. Furthermore, in the emerging field of active matter, self propulsion often leads to nematic order, both because of the broken symmetry in motion induced by the constituent particles, and because the elongated particles themselves promote liquid crystalline ordering \cite{Shankar_PRX2019,Shankar_sub2020}. Fruitful connections are being found with such disparate areas of Biology as rearrangements in confluent epithelial tissue \cite{Saw_Nat2017}, neural stem cell cultures \cite{Kawaguchi_Nat2017}, or cellular motors comprising microtubule bundles and kinesin complexes \cite{Sanchez_Nat2012}.

LCs are a \emph{meso-phase} of matter in which its ordered macroscopic state is \emph{between} a spatially disordered liquid, and a fully crystalline solid \cite{Virga_book1994}.  In their nematic phase, in which long ranged orientational order exists, the Landau-de Gennes (LdG) theory introduces a \emph{tensor-valued} function $\vQ$ to describe local order in the LC material.  In particular, the eigenframe of $\vQ$ yields information about the statistics of the distribution of LC molecule orientations.  The energy functional of $\vQ$ that describes the LC material involves both a bulk potential, and an elastic contribution involving the derivatives of $\vQ$.

Unlike the related description of a LC nematic phase in terms of a director, the
analysis based on $\vQ$ is generally limited to the so called one constant
approximation, appropriate for an elastically isotropic phase. In this case, the
Landau-de Gennes energy is supplemented by a gradient term of the form $L_{1}
|\nabla \vQ |^{2}$, where $L_{1}$ is the elastic constant. Inclusion of elastic
anisotropy requires gradient terms at least of third order in $\vQ$.
Unfortunately, at this order, the energy is known to become unbounded for any
choice of parameters \cite{Ball_MCLC2010,Bauman_CVPDE2016}. Therefore, in
principle, the requirement of a stable energy would imply consideration of terms
at least of fourth order in gradients. Since there are 22 possible fourth order
invariants allowed by symmetry \cite{Longa_LC1987}, the Landau-de Gennes theory
becomes overly complex for anisotropic systems.

This paper develops an alternative method to the LdG model that uses a special type of \emph{singular} bulk potential, the so-called Ball-Majumdar potential \cite{Ball_MCLC2010}.  This potential has the following desirable properties: (i), it can be derived from a microscopic interaction potential by using the tools of  statistical mechanics; and (ii), it enforces that $\vQ$ has physically permissible eigenvalues.  However, this choice of potential introduces some novel difficulties that are not present in the standard LdG model, chief among them is that in the implementation the energy as a function of $\vQ$ does not have a closed form, rather it needs to be evaluated entirely numerically. This is analogous to other self consistent field theories as applied, for example, to polymers \cite{Fredrickson_book2006}. Many numerical methods and implementations already exist for the standard LdG model, e.g. \cite{Bajc_JCP2016,Davis_SJNA1998,Ravnik_LC2009,Bartels_bookch2014,BorthNochettoWalker_NM2020,Lee_APL2002,Zhao_JSC2016}. However,  numerical methods for the Ball-Majumdar potential have been given only recently \cite{Schimming_PRE2020a,Schimming_PRE2020b}. In this paper we formalize this prior work, contrast its results with the LdG model in cases in which the latter fails, and also show the power of the method by computing defect configurations in three spatial dimensions.



\section{Liquid Crystal Theory}

We briefly review in this section the Landau-de Gennes theory for a nematic phase, as well as the more microscopic Maier-Saupe theory. Consider an anisotropic LC molecule which is uniaxial, with orientation described by the unit vector $\vp$. The Maier-Saupe potential between two molecules $i$ and $j$ is a contact interaction of the form $-\kappa \left( (\vp_{i} \cdot \vp_{j})^{2} -1/3) \right)$, where $\kappa$ is the interaction constant \cite{deGennes_book1995}. In the isotropic phase, the thermal average of $\vp$ is zero, while it is nonzero in the nematic phase. The Landau-de Gennes theory of a nematic phase is formulated instead in terms of a mesoscopic order parameter, the symmetric, traceless tensor $\vQ$. In the isotropic phase $\vQ = \vzero$. In the nematic phase, $\vQ$ is nonzero. A uniaxial nematic phase corresponds to two of the eigenvalues of $\vQ$ being equal, and a biaxial phase to the general case. Note that although the molecules themselves are uniaxial, the distribution of local orientations may itself be uniaxial or biaxial.

\subsection{Landau-de Gennes Theory}

Let $\symmtraceless$ be the set of symmetric, traceless $3 \times 3$ matrices.  Given a tensor-valued function $\vQ : \Om \to \symmtraceless$, where $\Om$ is a physical domain with Lipschitz boundary $\Gamma$, the free energy of the LdG model is defined as \cite{Mori_JJAP1999,Mottram_arXiv2014}:
\begin{equation}\label{eqn:Landau-deGennes_model_problem}
\begin{split}
  \ELdG [\vQ] &:= \iO \ELdGfunc (\vQ,\nabla \vQ) \, d\vx +
  \frac{1}{\bulkeps^2}
  \iO \Bulkfunc (\vQ) \, d\vx \\
&\qquad + \surfcoef \iG \LdGsurf(\vQ) \, dS(\vx) - \iO \LdGrhs(\vQ) \, d\vx,
\end{split}
\end{equation}
with
\begin{equation}\label{eqn:Landau-deGennes_energy_density}
\begin{split}
\ELdGfunc(\vQ,\nabla \vQ) &:= \frac{1}{2} \Big{(} \Li_{1} |\nabla \vQ|^2 + \Li_{2} |\nabla \cdot \vQ|^2 + \Li_{3} (\nabla \vQ)\tp \trid \nabla \vQ , \\
&\qquad\qquad + \Li_{4} \nabla \vQ \trid (\bflevi \cdot \vQ) + \Li_{*} \nabla \vQ \trid [(\vQ \cdot \nabla) \vQ] \Big{)},
\end{split}
\end{equation}
where $\{ \Li_{i} \}_{i=1}^{4}$, $\Li_*$, are material dependent elastic constants, and 
\begin{equation}\label{eqn:Landau-deGennes_invariants}
\begin{split}
|\nabla \vQ|^2 := (\partial_{k} Q_{ij})^2, \quad 
|\nabla \cdot \vQ|^2 := (\partial_{j} Q_{ij})^2, \quad (\nabla \vQ)\tp \trid \nabla \vQ &:= (\partial_{j} Q_{ik}) (\partial_{k} Q_{ij}), \\
\nabla \vQ \trid (\bflevi \cdot \vQ) := \levi_{j k l} Q_{ji} \partial_{l} Q_{ki}, \quad \nabla \vQ \trid [(\vQ \cdot \nabla) \vQ] &:= Q_{lk} (\partial_{l} Q_{ij}) (\partial_{k} Q_{ij}).
\end{split}
\end{equation}
We use the convention of summation over repeated indices.  Energies in Eq. \eqref{eqn:Landau-deGennes_model_problem} are made dimensionless by writing them in units of the temperature, $T$, while lengths are scaled by a characteristic length $\xi$.  The value of the dimensionless parameter $\bulkeps^2 \equiv L_{1} / (T \xi^2)$ determines the relative weight of the gradient dependent energy $\ELdGfunc(\vQ,\nabla \vQ)$ to the bulk potential $\Bulkfunc(\vQ)$ and thus determines $\xi$ \cite{Gartland_MMA2018}.  All five elastic constants can be related to the five independent constants of the Oseen-Frank model (i.e. $K_1$, $K_2$, $K_3$, $K_4$, and the twist $q_0$) \cite{Mori_JJAP1999,Mottram_arXiv2014}.  Indeed, $\Li_{4}$ accounts for twist and $\Li_{*}$ is needed to have five independent constants.  Note that taking $\Li_{i} = 0$, for $i=2,3,4$, and $\Li_{*} = 0$ gives the one constant LdG model.  More complicated models can also be considered \cite{Mottram_arXiv2014,deGennes_book1995,Sonnet_book2012}.  The bulk potential $\Bulkfunc (\vQ)$ is discussed in the next subsection.

The surface energy $\LdGsurf(\vQ)$, with parameter $\surfcoef \geq 0$, accounts for \emph{weak anchoring} of the LC (i.e. penalization of boundary conditions).  For example, a Rapini-Papoular type anchoring energy \cite{Barbero_JPF1986} can be considered:
\begin{equation}\label{eqn:Landau-deGennes_surf_energy}
\begin{split}
\LdGsurf (\vQ) = \frac{1}{2} \tr \left( \vQ - \vQgam \right)^2 \equiv \frac{1}{2} |\vQ - \vQgam|^2,
\end{split}
\end{equation}
where $\vQgam(\vx) \in \symmtraceless$ for all $\vx \in \Gm$.

The function $\LdGrhs(\cdot)$ accounts for interactions with external fields (e.g., an electric field).  For example, the energy density of a dielectric LC with fixed boundary potential is given by $-1/2 \, \vD \cdot \vE$ \cite{Walker_arXiv2018}, where the electric displacement $\vD$ is related to the electric field $\vE$ by the linear constitutive law \cite{Feynman_Lectures1964,deGennes_book1995,Biscari_CMT2007}:
\begin{equation}\label{eqn:LC_dielectric}
\vD = \etens \vE = \ebar \vE + \ea \vQ \vE, \quad \etens(\vQ) = \ebar \vI + \ea \vQ,
\end{equation}
where $\etens$ is the LC material's dielectric tensor and $\ebar$, $\ea$ are constitutive dielectric permittivities.  Thus, in the presence of an electric field, $\LdGrhs(\cdot)$ becomes
\begin{equation}\label{eqn:LC_electric_energy_functional}
\LdGrhs (\vQ) = -\frac{1}{2} \vD \cdot \vE = -\frac{1}{2} \left[ \ebar |\vE|^2 + \ea \vE \cdot \vQ \vE \right].
\end{equation}

\subsection{Landau-de Gennes bulk potential}\label{sec:bulk_potential}

The bulk potential $\Bulkfunc$ is a double-well type of function that is given by
\begin{equation}\label{eqn:Landau-deGennes_bulk_potential}
\begin{split}
\Bulkfunc (\vQ) = \bulkK - \frac{\bulkA}{2} \tr (\vQ^2) - \frac{\bulkB}{3} \tr (\vQ^3) + \frac{\bulkC}{4} \left( \tr (\vQ^2) \right)^2.
\end{split}
\end{equation}
Above, $\bulkA$, $\bulkB$, $\bulkC$ are material parameters such that $\bulkA$, $\bulkB$, $\bulkC$ are positive; $\bulkK$ is a convenient constant to ensure $\Bulkfunc \geq 0$.  Stationary points of $\Bulkfunc$ are either uniaxial or isotropic $\vQ$-tensors \cite{Majumdar_EJAM2010}.  

This potential was introduced to describe the vicinity of the isotropic-nematic phase transition, which is weakly first order. Therefore the eigenvalues of $\vQ$ are small. However, the same potential is used to describe systems deep inside the nematic phase, while not providing for any constraint on the eigenvalues. It is known that they can leave their physically admissible range in some circumstances. For example, consideration of an elastically anisotropic phase $K_{1} \neq K_{3}$ requires that $L_{*} \neq 0$. In this case, the energy $E[\vQ]$ is unbounded below for any choice of physical parameters \cite{Ball_MCLC2010,Bauman_CVPDE2016}, a divergence that is related to the absence of a constraint on the eigenvalues. The computational approach that we present here is precisely designed to remedy this problem.

\section{Self Consistent Mean Field Theory}\label{sec:MS_potential}

\subsection{Macroscopic order parameter}

We review the singular bulk potential introduced in \cite{Katriel_LC1986,Ball_MCLC2010}.  The goal is to have a bulk potential that correctly controls the eigenvalues of $\vQ \in \symmtraceless$, where $\symmtraceless$ is the set of symmetric, traceless $3 \times 3$ matrices.  Note that $\symmtraceless$ is spanned by a set of five basis matrices $\{ \basis^{k} \}_{k=1}^{5}$ \cite{Gartland_MCLC1991}.

The first step is to introduce a definition of the macroscopic order parameter (or mesoscopic field, under the assumption of local equilibrium), given by
\begin{equation}\label{eqn:defn_of_Q}
\begin{split}
\vQ = \iSp \left( \vp \otimes \vp - \frac{1}{3} \vI \right) \rho(\vp) \dA{\vp},
\end{split}
\end{equation}
where $\rho \in \Prob$ is the equilibrium probability distribution of the LC molecules given by statistical mechanics, i.e.
\begin{equation}\label{eqn:prop_meas_sphere}
\begin{split}
\Prob := \left\{ \rho \in L^1(\Sp; \R) \mid ~ \rho \geq 0, \quad \iSp \rho(\vp) \dA{\vp} = 1 \right\}.
\end{split}
\end{equation}
Note that $\vQ$ as defined is a thermal average. Therefore the minimization discussed in Sec. \ref{sec:SCFT_free_energy} at fixed $\vQ$ needs to be understood in a mean field sense. Note also that in the case of a non uniform configuration, we will assume that the same definition is valid so that an order parameter field $\vQ(\vx)$ is defined from the local distribution $\rho(\vp,\vx)$.

Equation \eqref{eqn:defn_of_Q} implies that the eigenvalues of $\vQ$, denoted $\lambda_{i} \equiv \lambda_{i}(\vQ)$, satisfy
\begin{equation}\label{eqn:Q_eigenvalue_bounds}
\begin{split}
-\frac{1}{3} \leq \lambda_{i} (\vQ) \leq \frac{2}{3}, ~\text{ for } i = 1,2,3, \qquad \sum_{i=1}^{3} \lambda_{i}(\vQ) = 0.
\end{split}
\end{equation}
In numerical work involving the Landau-de Gennes energy, equilibrium configurations of $\vQ$ are obtained by energy minimization, where the energy functional $E(\vQ)$ is \emph{independent} of any probability distribution of the underlying orientation of uniaxial molecules.  In other words, \eqref{eqn:Q_eigenvalue_bounds} is not guaranteed.  In contrast, the potential function defined below in Eq. \eqref{eqn:Maier-Saupe_pot} provides an energetic penalty so that the eigenvalues of $\vQ$ satisfy the bounds in \eqref{eqn:Q_eigenvalue_bounds}.

\subsection{Self-consistent free energy} 
\label{sec:SCFT_free_energy}

Let us define the entropy functional
\begin{equation}\label{eqn:entropy_functional}
\begin{split}
\entropy{\rho} = \iSp \rho(\vp) \ln \rho (\vp) \dA{\vp},
\end{split}
\end{equation}
and the intermolecular interaction kernel
\begin{equation}\label{eqn:intermolecular_kernel}
\begin{split}
K[\rho,\eta] = \iSp \iSp \left[ (\vp \cdot \vq)^2 - \frac{1}{3} \right] \rho(\vp) \eta(\vq) \dA{\vp} \dA{\vq},
\end{split}
\end{equation}
where $\rho$ and $\eta$ are two probability distribution functions in ${\cal P}$. The Maier-Saupe Potential is defined as
\begin{equation}\label{eqn:Maier-Saupe_pot}
\begin{split}
\MS{\rho} = T \entropy{\rho} - \kappa K[\rho,\rho],
\end{split}
\end{equation}
where $T > 0$ is temperature, and $\kappa > 0$ is a constant (we have omitted the Boltzmann constant $k_{B}$). With this definition, $I_{MS}$ reduces to the thermodynamic free energy when the distribution $\rho$ is the corresponding equilibrium probability distribution.

One, however, proceeds differently. Given a value of $\vQ$ (or locally, if a field $\vQ(\vx)$ is specified), we minimize $\MS{\rho}$ over the space of probability distribution functions with the condition that $\vQ$ is given by Eq.  (\ref{eqn:defn_of_Q}).

It is straightforward to write the interaction energy solely as a function of $\vQ$. We have
\begin{equation}\label{eqn:kernel_simplify_pf_1}
\begin{split}
K[\rho,\rho] &= \iSp \vq\tp \left( \iSp \left[ \vp \otimes \vp - \frac{1}{3} \vI \right] \rho(\vp) \dA{\vp} \right) \vq \rho(\vq) \dA{\vq} \\
&= \iSp \vq\tp \vQ \vq \rho(\vq) \dA{\vq} = \vQ \dd \iSp \left[ \vq \otimes \vq - \frac{1}{3} \vI \right] \rho(\vq) \dA{\vq} = \vQ \dd \vQ = |\vQ|^2,
\end{split}
\end{equation}
where we used the fact that $\vQ$ is traceless. Therefore, the energy term 
in the Maier-Saupe free energy of a given configuration $\vQ$ is 
simply $-\kappa  |\vQ|^{2}$.

The computation of the entropy for fixed $\vQ$ is more complex.  As in other field theories,  one needs to ``invert'' the relationship in \eqref{eqn:defn_of_Q},  i.e., given $\vQ$ find the corresponding $\rho$ that provides this value of $\vQ$ in equilibrium.  Of course, this is ill-posed, so we must impose some additional conditions. We use a mean field assumption, according to which $\rho$ minimizes $\MS{\rho}$ over the admissible set, and define the corresponding mean field free energy as
\begin{equation}\label{eqn:defn_bulk_pot}
\begin{split}
\Bulkfunc(\vQ) &:= \inf_{\rho \in \Admisrho{\vQ}} \MS{\rho}, \\
&= T \inf_{\rho \in \Admisrho{\vQ}} \entropy{\rho} - \kappa |\vQ|^2,
\end{split}
\end{equation}
where the admissible set is
\begin{equation}\label{eqn:defn_admis_set}
\begin{split}
\Admisrho{\vQ} := \left\{ \rho \in \Prob \mid \vQ = \iSp \left[ \vp \otimes \vp - \frac{1}{3} \vI \right] \rho(\vp) \dA{\vp} \right\}.
\end{split}
\end{equation}
Since $\vQ$ is fixed, we can focus on the entropy.  Define
\begin{equation}\label{eqn:defn_bulk_f}
\begin{split}
\fbulk (\vQ) :=
\begin{cases}
\inf_{\rho \in \Admisrho{\vQ}} \entropy{\rho}, & \text{if eigenvalues of } \vQ \text{ satisfy \eqref{eqn:Q_eigenvalue_bounds}}, \\
+\infty, & \text{else}.
\end{cases}
\end{split}
\end{equation}
Then
\begin{equation}\label{eqn:bulk_pot_explicit}
\begin{split}
\Bulkfunc(\vQ) &= T \fbulk(\vQ) - \kappa |\vQ|^2.
\end{split}
\end{equation}
Figure \ref{fig:BulkPot} shows plots of Eq. \eqref{eqn:bulk_pot_explicit} for $\vQ$ parameterized as
\begin{equation}\label{eqn:Q_param}
\begin{split}
\vQ &= S \left( \vn \otimes \vn - \frac{1}{3} \vI \right) + R \left( \vm \otimes \vm - \frac{1}{3} \vI \right)
\end{split}
\end{equation}
over the ``physical triangle'' as in \cite{Majumdar_EJAM2010}.  For the plots we set $\kappa / T = 4$.  As seen in the figure, when $\vQ$ gets close to the physical bounds the bulk potential diverges.  For this value of $\kappa / T$, there are three minima corresponding to uniaxial states where the director is $\vn$, $\vm$, or $\vn \times \vm$.  For $\kappa / T < 3.4049$ there is a single minimum at $\vQ = \mathbf{0}$ corresponding to the isotropic phase \cite{Schimming_PRE2020a}.

\begin{figure}
	\includegraphics[width = \textwidth]{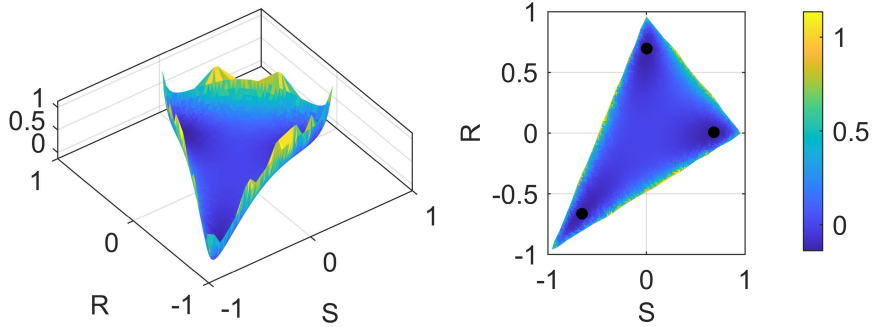}
	\caption{The bulk potential, Eq. \eqref{eqn:bulk_pot_explicit}, with $\vQ = S(\vn \otimes \vn - (1/3) \vI) + R(\vm \otimes \vm - (1/3) \vI)$ and $\kappa / T = 4$.  The energy goes to infinity as $\vQ$ approaches the physical bounds.  For the chosen value of $\kappa / T$ there are three minima that represent uniaxial states for each possible director direction which are highlighted with black dots.}
	\label{fig:BulkPot}
\end{figure}

\subsection{Properties}

Before proceeding with the numerical algorithm, and for completeness, we begin by summarizing a few preliminary results, see \cite{Ball_MCLC2010}.

\begin{lem}\label{lem:non-empty_admis}
For any $\vQ \in \symmtraceless$, such that $-1/3 < \lambda_{i}(\vQ) < 2/3$ (for $i=1,2,3$), the set $\Admisrho{\vQ}$ is non-empty.
\end{lem}
\begin{proof}
Given $\vQ \in \symmtraceless$, let $\vR$ be the orthogonal matrix that diagonalizes $\vQ$, i.e. $\vLambda = \vR\tp \vQ \vR$, where $(\vLambda)_{ii} = \lambda_{i}(\vQ)$, for $i=1,2,3$, where $2/3 > \lambda_{1} \geq 0$, $-1/3 < \lambda_{3} \leq 0$.  Now define the following (generalized) function (singular measure)
\begin{equation}\label{eqn:nonempty_admisQ_pf_1}
\begin{split}
\tilde{\rho} (\vp) = \sum_{k=1}^{3} \left( \lambda_{k} + \frac{1}{3} \right) \frac{\delta(\vp - \ve_{k}) + \delta(\vp + \ve_{k})}{2}, \quad \text{where} \quad \iSp \tilde{\rho} (\vp) \dA{\vp} = 1,
\end{split}
\end{equation}
where $\delta(\cdot - \va)$, for $\va \in \Sp$, is the Dirac delta function on $\Sp$ such that
\begin{equation}\label{eqn:nonempty_admisQ_pf_2}
\begin{split}
\iSp g(\vp) \delta(\vp - \va) \dA{\vp} = g(\va).
\end{split}
\end{equation}
Now let $X_{ij} := \iSp \left( \vp \otimes \vp - \frac{1}{3} \vI \right)_{ij} \tilde{\rho}(\vp) \dA{\vp}$, for $1 \leq i,j \leq 3$, and one can check that
\begin{equation}\label{eqn:nonempty_admisQ_pf_3}
\begin{split}
X_{ii} = \lambda_{i}, ~\text{ for } i=1,2,3, \quad \text{and} \quad X_{ij} = 0, ~\text{ for } i \neq j.
\end{split}
\end{equation}
In other words, we have
\begin{equation}\label{eqn:nonempty_admisQ_pf_4}
\begin{split}
\vLambda = \iSp \left( \vp \otimes \vp - \frac{1}{3} \vI \right) \tilde{\rho}(\vp) \dA{\vp},
\end{split}
\end{equation}
i.e. it satisfies the constraint.  Next, we replace $(\delta(\vp - \ve_{k}) + \delta(\vp + \ve_{k}))/2$ by a regularized version
\begin{equation}\label{eqn:nonempty_admisQ_pf_5}
\begin{split}
\phi_{k}^{\epsilon} (\vp) =
\begin{cases}
\frac{1}{2 |A_\epsilon^{1}|}, &\text{if } ~|\vp\cdot\ve_{k}| \geq 1 - \epsilon, \\
0, &\text{if } ~ |\vp\cdot\ve_{k}| < 1 - \epsilon,
\end{cases}
\end{split}
\end{equation}
where $\pm A_{\epsilon}^{k}$ denotes the spherical cap at $\pm \ve_{k}$ and $|A_\epsilon^{1}| = 2 \pi \epsilon$ is the area of one of the two spherical caps over which $\phi_{k}^{\epsilon} \neq 0$, with $\epsilon > 0$ small.  Now define a regularized version of \eqref{eqn:nonempty_admisQ_pf_1}:
\begin{equation}\label{eqn:nonempty_admisQ_pf_6}
\begin{split}
\tilde{\rho}^{\epsilon} (\vp) = a_0 \sum_{k=1}^{3} \left( \lambda_{k} + \frac{1}{3} + a_1 \right) \phi_{k}^{\epsilon} (\vp),
\end{split}
\end{equation}
where $a_0 > 0$ and $a_1$ are constants.  For convenience, define
\begin{equation}\label{eqn:nonempty_admisQ_pf_7}
\begin{split}
r_{i}^{k} = \frac{1}{2 |A_\epsilon^{1}|} \int_{+ A_{\epsilon}^{k} \cup - A_{\epsilon}^{k}} p_{i}^2 \dA{\vp}, \quad r_{i}^{i} = 1 - O(\epsilon), \forall i, \quad r_{i}^{k} = O(\epsilon^2), \text{ for } i \neq k,
\end{split}
\end{equation}
and note that $1 = \sum_{i=1} r_{i}^{k}$ for all $k = 1,2,3$, and symmetry implies
\begin{equation}\label{eqn:nonempty_admisQ_pf_7b}
\begin{split}
r_{i}^{k} = r_{k}^{i}, \quad \text{and} \quad r_{k}^{i} = r_{t}^{i} ~ \text{ for all } i, k, t ~\text{ distinct},
\end{split}
\end{equation}
which implies that $r_{k}^{k} + 2 r_{k}^{i} = 1$ whenever $i \neq k$. 
Now define $X_{ij}^{\epsilon} := \iSp \left( \vp \otimes \vp - \frac{1}{3} \vI \right)_{ij} \tilde{\rho}^{\epsilon}(\vp) \dA{\vp}$, for $1 \leq i,j \leq 3$, and compute:
\begin{equation}\label{eqn:nonempty_admisQ_pf_8}
\begin{split}
X_{ii}^{\epsilon} &= a_0 \sum_{k=1}^{3} \left( \lambda_{k} + \frac{1}{3} + a_1 \right) \frac{1}{2 |A_\epsilon^{1}|} \int_{+ A_{\epsilon}^{k} \cup - A_{\epsilon}^{k}} \left( p_{i}^2 - \frac{1}{3} \right) \dA{\vp} \\
&= a_0 \sum_{k=1}^{3} \left( \lambda_{k} + \frac{1}{3} + a_1 \right) \left( r_{i}^{k} - \frac{1}{3} \right) \\
&= a_0 \sum_{k \neq i} \left( \lambda_{k} + \frac{1}{3} + a_1 \right) \left( r_{i}^{k} - \frac{1}{3} \right) + a_0 \left( \lambda_{i} + \frac{1}{3} + a_1 \right) \left( r_{i}^{i} - \frac{1}{3} \right).
\end{split}
\end{equation}
Since $r_i^i = r_1^1$ for all $i$, and $r_i^k = r_1^2$ for all $i \neq k$, we continue to simplify \eqref{eqn:nonempty_admisQ_pf_8}:
\begin{equation}\label{eqn:nonempty_admisQ_pf_9}
\begin{split}
X_{ii}^{\epsilon} a_0^{-1} &= \left( r_{1}^{2} - \frac{1}{3} \right) \sum_{k \neq i} \left( \lambda_{k} + \frac{1}{3} + a_1 \right) + \left( \lambda_{i} + \frac{1}{3} + a_1 \right) \left( r_{1}^{1} - \frac{1}{3} \right) \\
&= \left( r_{1}^{2} - \frac{1}{3} \right) \left[ \frac{2}{3} - \lambda_{i} + 2 a_1 \right] + \left( \lambda_{i} + \frac{1}{3} + a_1 \right) \left( r_{1}^{1} - \frac{1}{3} \right) \\
&= \lambda_{i} \left( r_{1}^{1} - r_{1}^{2} \right) + \left( \frac{1}{3} + a_1 \right) \underbrace{\left( r_{1}^{1} + 2 r_{1}^{2} - 1 \right)}_{=0} \\
&= \lambda_{i} \left( r_{1}^{1} - r_{1}^{2} \right).
\end{split}
\end{equation}
Thus, letting $a_0 := \left( r_{1}^{1} - r_{1}^{2} \right)^{-1} = 1 + O(\epsilon)$, we get
\begin{equation}\label{eqn:nonempty_admisQ_pf_10}
\begin{split}
X_{ii}^{\epsilon} = \lambda_{i}, \quad \text{for} \quad i = 1,2,3.
\end{split}
\end{equation}
In addition, for $i \neq j$, we see that
\begin{equation}\label{eqn:nonempty_admisQ_pf_11}
\begin{split}
X_{ij}^{\epsilon} &= a_0 \sum_{k=1}^{3} \left( \lambda_{k} + \frac{1}{3} + a_1 \right) \frac{1}{2 |A_\epsilon^{1}|} \int_{+ A_{\epsilon}^{k} \cup - A_{\epsilon}^{k}} p_{i} p_{j} \dA{\vp} = 0,
\end{split}
\end{equation}
where the integral term drops by symmetry/cancellation.  Note that \eqref{eqn:nonempty_admisQ_pf_10} and \eqref{eqn:nonempty_admisQ_pf_11} hold for any value of $a_1$.  We must choose $a_1$ such that
\begin{equation}\label{eqn:nonempty_admisQ_pf_12}
\begin{split}
1 &= \iSp \tilde{\rho}^{\epsilon}(\vp) \dA{\vp} = a_0 \sum_{k=1}^{3} \left( \lambda_{k} + \frac{1}{3} + a_1 \right) = a_0 \left( 1 + 3 a_1 \right), \\
\Rightarrow \quad a_1 &= \frac{1}{3} \left( \frac{1}{a_0} - 1 \right) = O(\epsilon),
\end{split}
\end{equation}
where $a_1$ may be negative.  Since $\lambda_{k} + 1/3 > 0$, choosing $\epsilon > 0$ sufficiently small (but fixed), we see that $\lambda_{k} + (1/3) + a_1 > 0$ for $k = 1,2,3$.  Therefore, $\tilde{\rho}^{\epsilon}(\vp) \geq 0$ for all $\vp$, and $\tilde{\rho}^{\epsilon} \in \Prob$.  Moreover,
\begin{equation}\label{eqn:nonempty_admisQ_pf_14}
\begin{split}
\vLambda = \iSp \left( \vp \otimes \vp - \frac{1}{3} \vI \right) \tilde{\rho}^{\epsilon}(\vp) \dA{\vp},
\end{split}
\end{equation}
and $\tilde{\rho}^{\epsilon} \in \Admisrho{\vLambda}$.  
Finally, by rotating coordinates with $\vR$, and defining $\hat{\rho}^{\epsilon}(\vR \vp) := \tilde{\rho}^{\epsilon}(\vp)$, \eqref{eqn:nonempty_admisQ_pf_14} transforms into
\begin{equation}\label{eqn:nonempty_admisQ_pf_15}
\begin{split}
\vQ = \iSp \left( \hat{\vp} \otimes \hat{\vp} - \frac{1}{3} \vI \right) \hat{\rho}^{\epsilon}(\hat{\vp}) \dA{\hat{\vp}},
\end{split}
\end{equation}
where $\hat{\rho}^{\epsilon} \in \Admisrho{\vQ}$.  So $\Admisrho{\vQ}$ is non-empty.
\end{proof}

The following result lays out the main aspects of $\fbulk$ needed.
\begin{thm}\label{thm:results_for_fbulk}
Given $\vQ$ with $-1/3 < \lambda_{i}(\vQ) < 2/3$ (for $i=1,2,3$), there exists a unique minimizer $\rho^* \in \Admisrho{\vQ}$ to the optimization problem in \eqref{eqn:defn_bulk_f}.  In other words,
\begin{equation}\label{eqn:fbulk_value}
\begin{split}
\fbulk (\vQ) = \iSp \rho^*(\vp) \ln \rho^* (\vp) \dA{\vp},
\end{split}
\end{equation}
where
\begin{equation}\label{eqn:fbulk_density}
\begin{split}
\rho^*(\vp) = \frac{\exp \left( \vp\tp \vA \vp \right)}{\partition{\vA}}, \quad \partition{\vA} = \iSp \exp \left( \vp\tp \vA \vp \right) \dA{\vp},
\end{split}
\end{equation}
and $\vA \in \symmtraceless$ (symmetric, traceless) is the (unique) Lagrange multiplier for the constraint in \eqref{eqn:defn_admis_set}.  Moreover, $\vA$ satisfies the following non-linear equation, a requirement of mean field self-consistency:
\begin{equation}\label{eqn:multiplier_equation}
\begin{split}
\frac{1}{\partition{\vA}} \frac{\partial \partition{\vA}}{\partial \vA} \dd \vP = \vQ \dd \vP, \quad \text{for all } \vP \in \symmtraceless.
\end{split}
\end{equation}
\end{thm}
\begin{proof}
\textbf{Step 1.} We show that the minimization problem is well-posed.  From Lemma \ref{lem:non-empty_admis}, we know $\Admisrho{\vQ}$ is non-empty.  Moreover, the constraint in \eqref{eqn:defn_admis_set} is clearly convex, so $\Admisrho{\vQ}$ is a convex set.  In addition, $\entropy{\cdot}$ is a convex functional on $\Prob$, because $\rho \ln \rho$ is a (strictly) convex function of $\rho$.  Hence, $\entropy{\cdot}$ is weakly lower semi-continuous on $\Prob$.  So, by standard theory from the calculus of variations, there exists a minimizer $\rho^* \in \Admisrho{\vQ}$, and it is unique by convexity.

\textbf{Step 2.} Derive the Euler-Lagrange equations that characterize the minimizer.  We will mainly proceed formally, but this can be made more rigorous with similar arguments as in \cite[Ch. 8]{Evans:book}. In order to account for the constraint in $\Admisrho{\vQ}$, define the \emph{Lagrangian}
\begin{equation}\label{eqn:lagrangian_KKT_pf_1}
\begin{split}
\Lagr{\rho}{\vA} &:= \entropy{\rho} + \vA \dd \Constr{\rho} = \iSp \rho(\vp) \ln \rho (\vp) \dA{\vp} + \vA \dd \Constr{\rho}, \\
\Constr{\rho} &:= \vQ - \iSp \left( \vp \otimes \vp - \frac{1}{3} \vI \right) \rho(\vp) \dA{\vp} \in \symmtraceless,
\end{split}
\end{equation}
where $\vA$ is a constant matrix in $\symmtraceless$.  In order to account for the other constraints of $\rho$ being a probability measure, let us parameterize it:
\begin{equation}\label{eqn:lagrangian_KKT_pf_2}
\begin{split}
\rho(\vp) = \frac{e^{\omega (\vp)}}{\partition{\omega}}, \quad \text{where} \quad \partition{\omega} = \iSp e^{\omega (\vp)} \dA{\vp},
\end{split}
\end{equation}
where $\omega : \Sp \to \R \cup \{ -\infty \}$ is an ``arbitrary'' (measurable) function; thus, $\rho$ is a probability measure for any $\omega$.  We list some perturbation formulas that will be useful later.  Let $\omega_{\epsilon} = \omega + \epsilon \eta$, where $\epsilon > 0$ is small, and $\eta$ is an arbitrary measurable function (perturbation).  Then, standard variational calculus gives
\begin{equation}\label{eqn:lagrangian_KKT_pf_3}
\begin{split}
\delta_{\omega} (e^{\omega}) (\eta) &:= \frac{d}{d \epsilon} \Big{|}_{\epsilon=0} e^{\omega + \epsilon \eta} = \eta e^{\omega}, \quad \delta_{\omega} \partition{\omega} (\eta) = \iSp \eta(\vp) e^{\omega (\vp)} \dA{\vp}, \\
\delta_{\omega} \left( \frac{\omega}{\partition{\omega}} \right) (\eta) &= \frac{\eta}{\partition{\omega}} - \frac{\omega}{(\partition{\omega})^2} \delta_{\omega} \partition{\omega} (\eta) = \frac{1}{\partition{\omega}} \left( \eta - \omega \iSp \eta(\vp) \rho(\vp) \dA{\vp} \right) = \frac{1}{\partition{\omega}} \left( \eta - \omega \bar{\eta} \right), \\
\delta_{\omega} \left( \frac{e^{\omega}}{\partition{\omega}} \right) (\eta) &= \frac{\eta e^{\omega}}{\partition{\omega}} - \frac{e^{\omega}}{(\partition{\omega})^2} \delta_{\omega} \partition{\omega} (\eta) = \rho(\vp) \left( \eta(\vp) - \bar{\eta} \right),
\end{split}
\end{equation}
where $\bar{\eta} = \Eval{\rho}{\eta}$ (the expected value with respect to $\rho$).

\textbf{Step 3.} Next, we derive the KKT conditions for the optimal solution of the problem, in terms of $\omega$.  So, we instead form the Lagrangian \eqref{eqn:lagrangian_KKT_pf_1} in terms of $\omega$: $\Lagr{\omega}{\vA} := \entropy{\omega} + \vA \dd \Constr{\omega}$.  Computing the variation of the entropy, we have
\begin{equation}\label{eqn:lagrangian_KKT_pf_4}
\begin{split}
\delta_{\omega} \entropy{\omega} (\eta) &= \iSp \delta_{\omega} \left( \frac{e^{\omega}}{\partition{\omega}} \right) (\eta) \left[ \omega(\vp) - \ln \partition{\omega} \right] \dA{\vp} + \iSp \rho (\vp) \left[ \eta(\vp) - \frac{1}{\partition{\omega}} \delta_{\omega} \partition{\omega} (\eta) \right] \dA{\vp} \\
&= \iSp \rho(\vp) \left( \eta(\vp) - \bar{\eta} \right) \left[ \omega(\vp) + 1 - \ln \partition{\omega} \right] \dA{\vp} = \iSp \rho(\vp) \left( \eta(\vp) - \bar{\eta} \right) \omega(\vp) \dA{\vp},
\end{split}
\end{equation}
where the last equality is because of the definition of $\bar{\eta}$ and the fact that $+ 1 - \ln \partition{\omega}$ is a constant.  Note that $\rho$ satisfies \eqref{eqn:lagrangian_KKT_pf_2}.

Next, note that
\begin{equation}\label{eqn:lagrangian_KKT_pf_5}
\begin{split}
\vA \dd \Constr{\omega} = \vA \dd \vQ - \iSp \vp\tp \vA \vp \frac{e^{\omega (\vp)}}{\partition{\omega}} \dA{\vp}.
\end{split}
\end{equation}
Hence, the first variation of the constraint gives
\begin{equation}\label{eqn:lagrangian_KKT_pf_6}
\begin{split}
\delta_{\omega} (\vA \dd \Constr{\omega}) (\eta) &= - \iSp \vp\tp \vA \vp \, \delta_{\omega} \left( \frac{e^{\omega}}{\partition{\omega}} \right) (\eta) \dA{\vp} = - \iSp \vp\tp \vA \vp \, \rho(\vp) \left( \eta(\vp) - \bar{\eta} \right) \dA{\vp}.
\end{split}
\end{equation}

The first KKT condition is given by $\delta_{\omega} \Lagr{\omega}{\vA} (\eta) = 0$, for all admissible $\eta$.  Therefore, by the above calculations, we obtain
\begin{equation}\label{eqn:lagrangian_KKT_pf_7}
\begin{split}
0 &= \delta_{\omega} \Lagr{\omega}{\vA} (\eta) = \iSp \left[ \omega(\vp) - \vp\tp \vA \vp \right] \rho(\vp) \left( \eta(\vp) - \bar{\eta} \right) \dA{\vp},
\end{split}
\end{equation}
for all admissible $\eta$.  This implies that $\omega(\vp) - \vp\tp \vA \vp = c$, where $c$ is any constant.  So, from \eqref{eqn:lagrangian_KKT_pf_2}, we find
\begin{equation}\label{eqn:lagrangian_KKT_pf_8}
\begin{split}
\rho(\vp) = e^{c} \frac{\exp \left( \vp\tp \vA \vp \right)}{\partition{\omega}} = \frac{\exp \left( \vp\tp \vA \vp \right)}{\partition{\vA}},
\end{split}
\end{equation}
where the $e^{c}$ term cancels out.  Since the minimizer is unique, we have proven \eqref{eqn:fbulk_density}.  The second KKT condition simply recovers the constraint $\vzero = \partial \Lagr{\omega}{\vA} / \partial \vA = \Constr{\omega}$.

\textbf{Step 4.} Finally, the last step in the inversion is an equation to determine $\vA \in \symmtraceless$.  Starting from the relation $\partition{\vA} = \iSp \exp \left( \vp\tp \vA \vp \right) \dA{\vp}$, we first differentiate with respect $\vA$ but in the direction of general symmetric matrices, not necessarily trace free:
\begin{equation}\label{eqn:implicit_eqn_lagr_mult_pf_1}
\begin{split}
\frac{\partial \partition{\vA}}{\partial \vA} &= \iSp \frac{\partial}{\partial \vA} \left( \vp\tp \vA \vp \right) \exp \left( \vp\tp \vA \vp \right) \dA{\vp} = \iSp \left( \vp \otimes \vp \right) \exp \left( \vp\tp \vA \vp \right) \dA{\vp} \\
&= \iSp \left( \vp \otimes \vp - \frac{1}{3} \vI \right) \exp \left( \vp\tp \vA \vp \right) \dA{\vp} + \frac{1}{3} \vI \iSp \exp \left( \vp\tp \vA \vp \right) \dA{\vp} \\
&= \partition{\vA} \left[ \iSp \left( \vp \otimes \vp - \frac{1}{3} \vI \right) \rho (\vp) \dA{\vp} + \frac{1}{3} \vI \right] = \partition{\vA} \left[ \vQ + \frac{1}{3} \vI \right],
\end{split}
\end{equation}
where $\rho$ satisfies \eqref{eqn:lagrangian_KKT_pf_8}.  Thus, the multiplier $\vA$ satisfies the following equation
\begin{equation}\label{eqn:implicit_eqn_lagr_mult_pf_2}
\begin{split}
\frac{1}{\partition{\vA}} \frac{\partial \partition{\vA}}{\partial \vA} &= \vQ + \frac{1}{3} \vI.
\end{split}
\end{equation}
Dotting \eqref{eqn:implicit_eqn_lagr_mult_pf_2} with an arbitrary ``test'' function in $\symmtraceless$, we get \eqref{eqn:multiplier_equation}.
\end{proof}

We will also make use of these results for $\fbulk$. Note that the partition function \eqref{eqn:lagrangian_KKT_pf_2} is a single particle partition function obtained by integration over $\vp \in S^{2}$, but with {\em specified} values of the Lagrange multiplier $\vA$. The fact that it simply involves a quadratic form of $\vp$ originates from the form of the constrain \eqref{eqn:defn_of_Q}. In the mean field approximation considered, given $\vQ$ there is a unique $\vA$, defined by \eqref{eqn:implicit_eqn_lagr_mult_pf_2}, so that the corresponding $\rho$ in \eqref{eqn:lagrangian_KKT_pf_2} gives as average precisely $\vQ$.

The following result illustrates additional properties of $\fbulk(\vQ)$, 
including the simultaneous diagonalization of $\vQ$ and $\vA$, which can be 
useful for numerical implementation purposes.

\begin{cor}\label{cor:more_results_for_fbulk}
The function $\fbulk(\vQ)$ is a strictly convex function of $\vQ$.  In addition, any $\vQ \in \symmtraceless$, and the corresponding unique $\vA$ coming from solving the constrained minimization problem in Theorem \ref{thm:results_for_fbulk}, are diagonalized by the \emph{same} orthogonal matrix $\vR$, i.e.
\begin{equation}\label{eqn:simult_diag}
\vLambda = \vR\tp \vQ \vR, \quad \vPi = \vR\tp \vA \vR,
\end{equation}
where $\vLambda = \diag (\lambda_{1}, \lambda_{2}, \lambda_{3})$ and $\vPi = \diag (\pi_{1}, \pi_{2}, \pi_{3})$, where $\{ \pi_{i} \}_{i=1}^{3}$ are the eigenvalues of $\vA$.

Moreover, the Lagrange multiplier $\vA$ can be characterized as the optimal solution of the dual problem.  In other words, define
\begin{equation}\label{eqn:dual_convex_func}
\dual{\vA} := \ln \left( \partition{\vA} \right) - \vQ \dd \vA,
\end{equation}
where $\dual{\cdot} : \symmtraceless \to \R$ is a strictly convex function (but not uniformly strictly convex).  Then, the optimal Lagrange multiplier $\vA$ from Theorem \ref{thm:results_for_fbulk} is the unique minimizer of \eqref{eqn:dual_convex_func} over the set of symmetric, traceless matrices, i.e. an unconstrained minimization problem.
\end{cor}
\begin{proof}
\textbf{Convexity.}  Let $\vQ_{0}, \vQ_{1} \in \symmtraceless$ with eigenvalues satisfying \eqref{eqn:Q_eigenvalue_bounds}, and let $\rho_{i}$ be the minimizer in \eqref{eqn:fbulk_value} corresponding to $\vQ_{i}$, for $i=0,1$. Set $\vQ(t) = \vQ_{0} (1-t) + \vQ_{1} t$ for all $t \in [0,1]$, and also define $\rho(\cdot, t) := \rho_{0}(\cdot) (1-t) + \rho_{1} (\cdot) t$.  Then, since $\entropy{\rho}$ is a strictly convex functional of $\rho$, for all $0 < t < 1$ we have
\begin{equation}\label{eqn:convexity_fbulk_pf_1}
\begin{split}
\fbulk (\vQ(t)) &= \min_{\rho \in \Admisrho{\vQ(t)}} \entropy{\rho} \leq \entropy{\rho(\cdot,t)} < \entropy{\rho_{0}} (1 - t) + \entropy{\rho_{1}} t \\
&= \left(\min_{\rho \in \Admisrho{\vQ_{0}}} \entropy{\rho} \right) (1-t) + \left(\min_{\rho \in \Admisrho{\vQ_{1}}} \entropy{\rho} \right) t = \fbulk (\vQ_{0}) (1-t) + \fbulk (\vQ_{1}) t,
\end{split}
\end{equation}
which verifies the strict convexity of $\fbulk(\cdot)$.

\textbf{Simultaneous diagonalization.}  Given $\vQ$, let $\vA$ and $\rho$ be the optimal solution of the constrained minimization problem, and let $\vR$ be the orthogonal matrix such that $\vPi = \vR\tp \vA \vR$ where $\vPi$ is a diagonal matrix. Then, using the change of variable $\vp = \vR \vq$, we get
\begin{equation}\label{eqn:simult_diag_pf_1}
\begin{split}
\partition{\vA} \vQ &= \iSp \left( \vp \vp\tp - \frac{1}{3} \vI \right) \exp \left( \vp\tp \vA \vp \right) \dA{\vp} = \iSp \left( \vR \vq \vq\tp \vR\tp - \frac{1}{3} \vR \vR\tp \right) \exp \left( \vq\tp \vR\tp \vA \vR \vq \right) \dA{\vq} \\
&= \vR \left[ \iSp \left( \vq \vq\tp - \frac{1}{3} \vI \right) \exp \left( \sum_{i=1}^{3} \pi_{i} q_{i}^2 \right) \dA{\vq} \right] \vR\tp,
\end{split}
\end{equation}
i.e.
\begin{equation}\label{eqn:simult_diag_pf_2}
\begin{split}
\vR\tp \vQ \vR &= \frac{1}{\partition{\vA}} \iSp \left( \vq \vq\tp - \frac{1}{3} \vI \right) \exp \left( \sum_{i=1}^{3} \pi_{i} q_{i}^2 \right) \dA{\vq}.
\end{split}
\end{equation}
Now, note that for $i \neq j$, there holds
\begin{equation}\label{eqn:simult_diag_pf_3}
\begin{split}
0 = \iSp \underbrace{q_{i} q_{j}}_{\text{odd}} \underbrace{\exp \left( \sum_{i=1}^{3} \pi_{i} q_{i}^2 \right)}_{\text{even}} \dA{\vq}.
\end{split}
\end{equation}
That means $\vR\tp \vQ \vR$ must be diagonal.  Since matrix diagonalization (with orthogonal matrices) is unique, $\vR\tp \vQ \vR \equiv \vLambda$.

\textbf{The dual minimization problem.}  Let $\vA^*$ be the optimal Lagrange multiplier from \eqref{eqn:multiplier_equation}; it is clear that $\vA^*$ solves the first order condition for \eqref{eqn:dual_convex_func}:
\begin{equation}\label{eqn:dual_convex_func_pf_1}
0 = \frac{\partial \dual{\vA}}{\partial \vA} \dd \vP = \frac{1}{\partition{\vA}} \frac{\partial \partition{\vA}}{\partial \vA} \dd \vP - \vQ \dd \vP, \quad \text{for all } \vP \in \symmtraceless.
\end{equation}
Next, we compute the Hessian of $\dual{\vA}$ for any $\vA \in \symmtraceless$:
\begin{equation}\label{eqn:dual_convex_func_pf_2}
\begin{split}
\vT \dd \frac{\partial^2 \dual{\vA}}{\partial \vA^2} \dd \vP &= \frac{1}{\partition{\vA}} \vT \dd \frac{\partial^2 \partition{\vA}}{\partial \vA^2} \dd \vP - \frac{1}{(\partition{\vA})^2} \left( \frac{\partial \partition{\vA}}{\partial \vA} \dd \vT \right) \left( \frac{\partial \partition{\vA}}{\partial \vA} \dd \vP \right),
\end{split}
\end{equation}
where
\begin{equation}\label{eqn:dual_convex_func_pf_3}
\begin{split}
\frac{1}{\partition{\vA}} \frac{\partial \partition{\vA}}{\partial \vA} \dd \vP &= \iSp \vp\tp \vP \vp \, \rho (\vp) \dA{\vp} = \Eval{\rho}{\vp\tp \vP \vp}, \\
\frac{1}{\partition{\vA}} \vT \dd \frac{\partial^2 \partition{\vA}}{\partial \vA^2} \dd \vP &= \iSp \left(\vp\tp \vT \vp \right) \left(\vp\tp \vP \vp \right) \rho (\vp) \dA{\vp} = \Eval{\rho}{\left(\vp\tp \vT \vp \right) \left(\vp\tp \vP \vp \right)},
\end{split}
\end{equation}
for all $\vP, \vT \in \symmtraceless$; note that $\Eval{\rho}{\cdot}$ is the expected value with respect to $\rho$, which is determined by $\vA$.  
To show strict convexity, we must verify that $\partial_{\vA}^2 \dual{\vA}$ is positive definite for all $\vA$: 
\begin{equation}\label{eqn:dual_convex_func_pf_4}
\begin{split}
\vP \dd \frac{\partial^2 \dual{\vA}}{\partial \vA^2} \dd \vP &= \Eval{\rho}{\left(\vp\tp \vP \vp \right)^2} - \left( \Eval{\rho}{\vp\tp \vP \vp} \right)^2 = \Eval{\rho}{\left( \vp\tp \vP \vp - \Eval{\rho}{\vp\tp \vP \vp} \right)^2},
\end{split}
\end{equation}
which is the covariance of $\vp\tp \vP \vp$ with respect to $\rho$ (which depends on $\vA$) and is always positive semi-definite.  If \eqref{eqn:dual_convex_func_pf_4} were identically zero, then that would imply that some marginal distribution of $\rho$ is a Dirac delta.  But this is not possible given the form of $\rho$ in \eqref{eqn:lagrangian_KKT_pf_8} so long as $\vA$ is finite.  Therefore, $\frac{\partial^2 \dual{\vA}}{\partial \vA^2}$ is positive definite for all $\vA \in \symmtraceless$ such that $|\vA| < \infty$. This means that $\dual{\cdot}$ is strictly convex.
\end{proof}

\section{Minimizing the Landau-de Gennes Energy}

The free energy minimization of the self consistent free energy \eqref{eqn:bulk_pot_explicit} shares many elements with minimization procedures of the Landau-de Gennes free energy. We summarize known results concerning the later here, and emphasize the differences with the proposed method.

\subsection{Existence of a Minimizer}


The admissible space for $\vQ$ when seeking a minimizer is
\begin{equation}\label{eqn:Landau-deGennes_function_space}
\begin{split}
\LdGspace{\vP} := \left\{ \vQ \in H^{1}(\Om; \symmtraceless) \mid \vQ |_{\Gmdir} = \vP \right\},
\end{split}
\end{equation}
where $\vP \in H^{1}(\Om;\symmtraceless)$.  Note that $\symmtraceless$ is spanned by a set of five orthonormal basis matrices $\{ \basis^{k} \}_{k=1}^{5}$. The set $\Gmdir \subset \Gm$ is where strong anchoring is imposed, i.e. $\vQ |_{\Gm} = \vQdir \in H^{1}(\Gm; \symmtraceless)$, where $\Bulkfunc(\vQdir(\vx)) < \infty$ for all $\vx \in \Gm$.  
The weak anchoring function $\vQgam$ is taken in $L^{2}(\Gm;\symmtraceless)$, with $\Bulkfunc(\vQgam(\vx)) < \infty$ for all $\vx \in \Gm$.  
The minimization problem for the LdG energy functional \eqref{eqn:Landau-deGennes_model_problem} is
\begin{equation}\label{eqn:LdG_min_problem}
\begin{split}
\min_{ \vQ \in \LdGspace{\vQdir}} \ELdG[\vQ],
\end{split}
\end{equation}

Existence of a minimizer requires the energy to be bounded from below.  The following theorem \cite[Lem. 4.1]{Davis_SJNA1998} establishes this result for  the $\Li_{1}$, $\Li_{2}$, and $\Li_{3}$ terms only.
\begin{thm}\label{thm:Gartland_coercivity}
Let $\aform{\cdot}{\cdot} : H^1(\Om;\symmtraceless) \times H^1(\Om;\symmtraceless) \to \R$ be the symmetric bilinear form defined by
\begin{equation}\label{eqn:LdG_form}
\begin{split}
\aform{\vP}{\vT} = \iO \Li_{1} \nabla \vP \trid \nabla \vT + \Li_{2} (\nabla \cdot \vP) \cdot (\nabla \cdot \vT) + \Li_{3} (\nabla \vP)\tp \trid \nabla \vT \, d\vx.
\end{split}
\end{equation}
Then $\aform{\cdot}{\cdot}$ is bounded.  If $\Li_{1}$, $\Li_{2}$, $\Li_{3}$ satisfy
\begin{equation}\label{eqn:Gartland_coercivity}
\begin{split}
0 < \Li_{1}, \quad - \Li_{1} < \Li_{3} < 2 \Li_{1}, \quad - \frac{3}{5} \Li_{1} - \frac{1}{10} \Li_{3} < \Li_{2},
\end{split}
\end{equation}
then there is a constant $C > 0$ such that $\aform{\vP}{\vP} \geq C | \vP |^2_{H^1(\Om)}$ for all $\vP \in H^1(\Om)$.  Moreover, if $|\Gmdir| > 0$, then there is a constant $C' > 0$ such that $\aform{\vP}{\vP} \geq C' \| \vP \|^2_{H^1(\Om)}$ for all $\vP \in \LdGspace{\vzero}$.
\end{thm}

We also have the bilinear form $\bform{\cdot}{\cdot} : H^1(\Om;\symmtraceless) \times H^1(\Om;\symmtraceless) \to \R$ and trilinear form $\cform{\cdot}{\cdot}{\cdot} :H^1(\Om;\symmtraceless) \times H^1(\Om;\symmtraceless) \times H^1(\Om;\symmtraceless) \to \R$ accounting for the $\Li_{4}$ and $\Li_{*}$ terms:
\begin{equation}\label{eqn:twist_form}
\begin{split}
\bform{\vP}{\vT} = \frac{\Li_{4}}{2} \iO \nabla \vP \trid (\bflevi \cdot \vT) + \nabla \vT \trid (\bflevi \cdot \vP) \, d\vx,
\end{split}
\end{equation}
\begin{equation}\label{eqn:cubic_form}
\begin{split}
\cform{\vT}{\vP}{\vQ} = \frac{\Li_{*}}{2} \iO \Big{\{} &\nabla \vP \trid [(\vT \cdot \nabla) \vQ] + \nabla \vP \trid [(\vQ \cdot \nabla) \vT] \\
     + &\nabla \vT \trid [(\vP \cdot \nabla) \vQ] + \nabla \vT \trid [(\vQ \cdot \nabla) \vP] \\
     + & \nabla \vQ \trid [(\vP \cdot \nabla) \vT] + \nabla \vQ \trid [(\vT \cdot \nabla) \vP] \Big{\}} d\vx.
\end{split}
\end{equation}
Next, consider the following sub-part of the energy \eqref{eqn:Landau-deGennes_model_problem}:
\begin{equation}\label{eqn:subpart_energy}
\begin{split}
	\widetilde{\ELdG} [\vQ] &:= \iO \ELdGfunc (\vQ,\nabla \vQ) \, d\vx + \frac{1}{\bulkeps^2} \iO \Bulkfunc (\vQ) \, d\vx, \\
	&\equiv \frac{1}{2} \aform{\vQ}{\vQ} + \frac{1}{2} \bform{\vQ}{\vQ} + \frac{1}{6} \cform{\vQ}{\vQ}{\vQ} + \frac{1}{\bulkeps^2} \iO \Bulkfunc (\vQ) \, d\vx.
\end{split}
\end{equation}

Combining Theorem \ref{thm:Gartland_coercivity} with the form of the energy in
\eqref{eqn:Landau-deGennes_model_problem} and other basic results (see
\cite[Lem. 4.2, Thm. 4.3]{Davis_SJNA1998} for instance) we arrive at the
following result. 

\begin{thm}[existence of a minimizer]\label{thm:exist_min}
Suppose that $\vQdir$ and $\Gmdir$ are defined as above and that $\LdGrhs$ is a bounded linear functional on $\LdGspace{\vQdir}$.  
Let $\widetilde{\ELdG}$ be given by \eqref{eqn:subpart_energy}, where $\Bulkfunc$ is given by \eqref{eqn:bulk_pot_explicit}.  Furthermore, assume $T > 0$, let $\Li_{4}$ be bounded, and assume that
\begin{equation}\label{eqn:cubic_coeff_bound}
	\widetilde{\Li}_{1} := \Li_{1} - \max \left( \frac{\Li_{*}}{3}, -\frac{3}{2} \Li_{*} \right),
\end{equation}
and $\widetilde{\Li}_{1}$, $\Li_{2}$, $\Li_{3}$ satisfy \eqref{eqn:Gartland_coercivity} with $\Li_{1}$ replaced by $\widetilde{\Li}_{1}$.  
Then $\widetilde{\ELdG}$ has a minimizer in the space $\LdGspace{\vQdir}$, whose eigenvalues are \emph{strictly within} the physical limits \eqref{eqn:Q_eigenvalue_bounds} almost everywhere.  
Furthermore, $\ELdG$ in \eqref{eqn:Landau-deGennes_model_problem}, with $\Bulkfunc$ given by \eqref{eqn:bulk_pot_explicit}, has a minimizer in $\LdGspace{\vQdir}$, whose eigenvalues are \emph{strictly within} the physical limits \eqref{eqn:Q_eigenvalue_bounds} almost everywhere. 
\end{thm}
\begin{proof}
When $\Li_{4} = \Li_{*} = 0$, the result follows from \cite[Lem. 4.2, Thm. 4.3]{Davis_SJNA1998}.  Otherwise, consider the case where $\Li_{*}$ is positive (the negative case is similar).  Consider the constrained admissible set
\begin{equation}\label{eqn:exist_min_pf_0}
	\AdmisLdG := \{ \vQ \in \LdGspace{\vQdir} ~ | -1/3 \leq \lambda_{i} (\vQ) \leq 2/3, \text{ for } i = 1, 2, 3 \},
\end{equation}
and note that $\AdmisLdG$ is a closed, convex set.  Since the minimum eigenvalue of $\Li_{1} \vI + \Li_{*} \vQ$ (on $\AdmisLdG$) is $\widetilde{\Li}_{1}$, using Theorem \ref{thm:Gartland_coercivity}, we have that $\widetilde{\ELdG}$ satisfies the bound
\begin{equation}\label{eqn:exist_min_pf_1}
\begin{split}
	\widetilde{\ELdG}[\vQ] &\geq \frac{1}{2} \iO \widetilde{\Li}_{1} |\nabla \vQ|^2 + \Li_{2} (\nabla \cdot \vQ)^2 + \Li_{3} (\nabla \vQ)\tp \trid \nabla \vQ \, d\vx \\
	& + \frac{\Li_{4}}{2} \iO \nabla \vQ \trid (\bflevi \cdot \vQ) \, d\vx + \frac{1}{\bulkeps^2} \iO \Bulkfunc (\vQ) \, d\vx \\
	&\geq \frac{C}{2} \iO |\nabla \vQ|^2 \, d\vx + \frac{\Li_{4}}{2} \iO \nabla \vQ \trid (\bflevi \cdot \vQ) \, d\vx + \frac{1}{\bulkeps^2} \iO \Bulkfunc (\vQ) \, d\vx,
\end{split}
\end{equation}
for all $\vQ \in \AdmisLdG$, for some constant $C > 0$.  Furthermore, one can show
\begin{equation}\label{eqn:exist_min_pf_2}
\begin{split}
\widetilde{\ELdG}[\vQ] &\geq \frac{1}{2} \left( C - \zeta_0 \right) \iO |\nabla \vQ|^2 \, d\vx - \frac{C'}{\zeta_0} \iO |\vQ|^2 \, d\vx + \frac{1}{\bulkeps^2} \iO \Bulkfunc (\vQ) \, d\vx,
\end{split}
\end{equation}
for any $\zeta_0 > 0$ where $C' > 0$ is some bounded constant.  Choosing $\zeta_0 = C / 2$, we get
\begin{equation}\label{eqn:exist_min_pf_3}
\begin{split}
\widetilde{\ELdG}[\vQ] &\geq \frac{C}{4} \iO |\nabla \vQ|^2 \, d\vx + \frac{1}{\bulkeps^2} \iO \widetilde{\Bulkfunc} (\vQ) \, d\vx,
\end{split}
\end{equation}
where $\widetilde{\Bulkfunc} (\vQ) := T \fbulk(\vQ) - (\kappa + 2 \bulkeps^2 C'/C) |\vQ|^2$.  Thus, $\widetilde{\ELdG}[\vQ]$ is clearly bounded below by a coercive energy on $\AdmisLdG$.  By standard calculus of variations \cite{Braides_book2002,Jost_Book1998}, there exists a minimizer, $\widetilde{\vQ}$, of $\widetilde{\ELdG}[\cdot]$ in $\AdmisLdG$.  Moreover, $\fbulk(\widetilde{\vQ}) < \infty$ almost everywhere, meaning the eigenvalues of $\widetilde{\vQ}$ are strictly within the physical limits almost everywhere.  The same holds true for $\ELdG[\cdot]$.
\end{proof}

\subsection{Gradient Flow}
\label{sec:gradient_flow}

We look for an energy minimizer using a gradient flow strategy \cite{Lee_APL2002,Ravnik_LC2009,Zhao_JSC2016,Bajc_JCP2016} applied to the energy \eqref{eqn:Landau-deGennes_model_problem}.  
Let $t$ represent ``time'' and suppose that $\vQ \equiv \vQ(\vx,t)$ satisfies an evolution equation such that $\lim_{t \to \infty} \vQ(\cdot,t) =: \vQ_{*}$ is a local minimizer of $\ELdG$, where $\vQ(\vx,0) = \vQinit$, and $\vQinit \in \LdGspace{\vQdir}$ is the initial guess for the minimizer.  The tensor $\vQ(\cdot,t)$ evolves according to the following $L^2(\Om)$ gradient flow:
\begin{equation}\label{eqn:LdG_L2_grad_flow}
\begin{split}
\ipOm{\partial_{t} \vQ(\cdot,t)}{\vP} = -\delta_{\vQ} \ELdG[\vQ;\vP], \quad \forall \vP \in \LdGspace{\vzero},
\end{split}
\end{equation}
where $\ipOm{\cdot}{\cdot}$ is the $L^2$ inner product over $\Om$. Formally, the solution of \eqref{eqn:LdG_L2_grad_flow} will converge to $\vQ_{*}$.

\begin{remark}\label{rem:tensor-valued_Allen-Cahn}
If $\Bulkfunc$ is the Landau-de Gennes bulk potential in \eqref{eqn:Landau-deGennes_bulk_potential}, then \eqref{eqn:LdG_L2_grad_flow} is essentially a tensor-valued Allen-Cahn equation.  By the standard theory of parabolic PDEs \cite{Friedman_book2008,Krylov_book2008}, it has a unique solution.  The same result also holds when $\Bulkfunc$ is the singular bulk potential.
\end{remark}

We use a numerical scheme for approximating \eqref{eqn:LdG_L2_grad_flow} by first discretizing in time by minimizing movements \cite{DeGiorgi_collection2006}. Let $\vQ_{k} (\vx) \approx \vQ(\vx,k \dt)$, where $\dt > 0$ is a finite time-step, and $k$ is the time index.  Then \eqref{eqn:LdG_L2_grad_flow} becomes a sequence of variational problems.  Given $\vQ_{k}$, find $\vQ_{k+1} \in \LdGspace{\vQdir}$ such that
\begin{equation}\label{eqn:LdG_L2_grad_flow_semi-discrete}
\begin{split}
\ipOm{\frac{\vQ_{k+1} - \vQ_{k}}{\dt}}{\vP} = -\delta_{\vQ} \ELdG[\vQ_{k+1};\vP], \quad \forall \vP \in \LdGspace{\vzero},
\end{split}
\end{equation}
which is equivalent to
\begin{equation}\label{eqn:LdG_L2_grad_flow_semi-discrete_min_movement}
\begin{split}
\vQ_{k+1} = \argmin_{\vQ \in \LdGspace{\vQdir}} F(\vQ), \quad F(\vQ) := \frac{1}{2 \dt} \| \vQ - \vQ_{k} \|_{L^2(\Om)}^2 + \ELdG[\vQ],
\end{split}
\end{equation}
and yields the useful property $F(\vQ_{k+1}) \leq F(\vQ_{k})$.  However, \eqref{eqn:LdG_L2_grad_flow_semi-discrete} is a \emph{fully-implicit} equation and requires an iterative solution because of the non-linearities in $\ELdGfunc(\vQ,\nabla \vQ)$ and $\Bulkfunc(\vQ)$.  As in the Landau-de Gennes case, $\Bulkfunc(\vQ)$ is non-convex \cite{Schimming_PRE2020a}, so that we adopt a \emph{convex splitting} method \cite{Wise_SJNA2009,Zhao_JSC2016,Xu_CMAME2019}.  Setting $\bulkimp (\vQ) = T \fbulk(\vQ)$ and $\bulkexp (\vQ) = \kappa |\vQ|^2$, we see that \eqref{eqn:bulk_pot_explicit} already has the form of a convex splitting:
\begin{equation}\label{eqn:LdG_convex_split}
\begin{split}
\Bulkfunc (\vQ) &\equiv \bulkimp (\vQ) - \bulkexp (\vQ),
\end{split}
\end{equation}
i.e. $\bulkimp$ and $\bulkexp$ are convex functions of $\vQ$.

In computing \eqref{eqn:LdG_L2_grad_flow_semi-discrete}, we treat $\bulkimp$ implicitly and $\bulkexp$ explicitly. Therefore, \eqref{eqn:LdG_L2_grad_flow_semi-discrete} becomes the following.  Given $\vQ_{k}$, find $\vQ_{k+1} \in \LdGspace{\vQdir}$ such that
%
%
\begin{equation}\label{eqn:LdG_L2_grad_flow_semi-implicit}
\begin{split}
&\ipOm{\frac{\vQ_{k+1} - \vQ_{k}}{\dt}}{\vP} + \aform{\vQ_{k+1}}{\vP} + \bform{\vQ_{k+1}}{\vP} + \cform{\vQ_{k+1}}{\vP}{\vQ_{k+1}} \\
&\frac{1}{\bulkeps^2} \iO \frac{\partial \bulkimp (\vQ_{k+1})}{\partial \vQ} \dd \vP \, d\vx + \surfcoef \iG \frac{\partial \LdGsurf(\vQ_{k+1})}{\partial \vQ} \dd \vP \, dS(\vx) \\
&= \frac{1}{\bulkeps^2} \iO \frac{\partial \bulkexp (\vQ_{k})}{\partial \vQ} \dd \vP \, d\vx + \iO \frac{\partial \LdGrhs(\vQ_{k})}{\partial \vQ} \dd \vP \, d\vx, \quad \forall \vP \in \LdGspace{\vzero},
\end{split}
\end{equation}
where the right-hand-side of \eqref{eqn:LdG_L2_grad_flow_semi-implicit} is completely explicit.  We then apply Newton's method to solving \eqref{eqn:LdG_L2_grad_flow_semi-implicit}.

Next, we approximate \eqref{eqn:LdG_L2_grad_flow_semi-implicit} by a finite element method, so we introduce some basic notation and assumptions in that regard.  We assume that $\Om \subset \R^{3}$ is discretized by a conforming shape regular triangulation $\Tk_{h} = \{ T_{i} \}$ consisting of simplices, i.e. we define $\Om_{h} := \cup_{T \in \Tk_{h}} T$.  Furthermore, we define the space of continuous piecewise linear functions on $\Om_{h}$:
\begin{equation}\label{eqn:std_P1_FE_space}
\begin{split}
\M_{h} (\Om_{h}) := \left\{ v \in C^{0} (\Om_{h}) \mid v |_{T} \in \Pk_{1} (T), ~\forall T \in \Tk_{h} \right\},
\end{split}
\end{equation}
where $\Pk_{k} (T)$ is the space of polynomials of degree $\leq k$ on $T$.

We discretize \eqref{eqn:LdG_L2_grad_flow_semi-implicit} by a $\Pk_{1}$ approximation of the $\vQ$ variable denoted $\vQ_{h}$.  To this end, define
\begin{equation}\label{eqn:Q-tensor_FE_space}
\begin{split}
	\Sh (\Om_{h}) := \left\{ \vP \in C^{0} (\Om_{h};\symmtraceless) \mid \vP = \sum_{i=1}^5 q_{i,h} \basis^{i}, ~ q_{i,h} \in \M_{h} (\Om_{h}), 1 \leq i \leq 5 \right\},
\end{split}
\end{equation}
and let $\vQ_{h} \in \Sh (\Om_{h})$.  Thus, $\vQ_{h} = \sum_{i=1}^5 q_{i,h} \basis^{i}$, and $\vQ_{h} \in H^1(\Om; \symmtraceless)$.

The fully discrete $L^2$-gradient flow now follows from \eqref{eqn:LdG_L2_grad_flow_semi-implicit}, which we explicitly state.  Given $\vQ_{h,k}$, find $\vQ_{h,k+1} \in \Sh (\Om_{h}) \cap \LdGspace{\interp \vQdir}$, where $\interp$ denotes the Lagrange interpolation operator, such that
\begin{equation}\label{eqn:LdG_L2_grad_flow_FE_approx}
\begin{split}
&\ipOm{\frac{\vQ_{h,k+1} - \vQ_{h,k}}{\dt}}{\vP} + \aform{\vQ_{h,k+1}}{\vP} + \bform{\vQ_{h,k+1}}{\vP} + \cform{\vQ_{h,k+1}}{\vP}{\vQ_{h,k+1}} \\
&\frac{1}{\bulkeps^2} \iO \interp \left( \frac{\partial \bulkimp (\vQ_{h,k+1})}{\partial \vQ} \right) \dd \vP \, d\vx + \surfcoef \iG \frac{\partial \LdGsurf(\vQ_{h,k+1})}{\partial \vQ} \dd \vP \, dS(\vx) \\
&= \frac{1}{\bulkeps^2} \iO \frac{\partial \bulkexp (\vQ_{h,k})}{\partial \vQ} \dd \vP \, d\vx + \iO \frac{\partial \LdGrhs(\vQ_{h,k})}{\partial \vQ} \dd \vP \, d\vx, \quad \forall \vP \in \Sh (\Om_{h}) \cap \LdGspace{\vzero}.
\end{split}
\end{equation}
We iterate this procedure until some stopping criteria is achieved.  As is the case with the Landau-de Gennes model, solving \eqref{eqn:LdG_L2_grad_flow_FE_approx} at each time-step requires Newton's method, i.e. we must compute the gradient and Hessian of the energy. In particular, we need to compute the gradient and Hessian of the singular bulk potential $\bulkimp \equiv T \fbulk$, as we detail in Section \ref{sec:eval_singular_bulk}.


\section{Evaluating the Singular Bulk Potential}\label{sec:eval_singular_bulk}

As in other self consistent mean field theories, the main difficulty with the
method is that \emph{no explicit formula} for $\fbulk(\vQ)$ is available.
Instead, one has to solve the
mean field self-consistency equation \eqref{eqn:multiplier_equation} numerically
for a given $\vQ$. In the application of Newton's method to the solution of
\eqref{eqn:LdG_L2_grad_flow_FE_approx}, we must compute $\partial \bulkimp (\vQ)
/\partial \vQ$ and $\partial^2 \bulkimp (\vQ) /\partial \vQ^2$, evaluated at the
current guess of the solution $\vQ_{h,k+1}$, where $\bulkimp \equiv T \fbulk$.
Since $\vQ_{h,k+1} = \vQ_{h,k+1}(\vx)$ is spatially varying, this can
potentially be very expensive to compute.  However, note the presence of the
Lagrange interpolation operator $\interp$ in
\eqref{eqn:LdG_L2_grad_flow_FE_approx}, i.e. see the term $\interp \left(
\partial \bulkimp (\vQ_{h,k+1}) /\partial \vQ \right)$.  Hence, $\fbulk$, and
its derivatives, need only be evaluated at the finite element degrees-of-freedom
(or nodes) of the mesh.  Moreover, the computation at each node is completely
independent of all other nodes, i.e. it is \emph{embarrassingly parallel}.
Therefore, the numerical implementation of the singular bulk potential is
completely tractable.

There are two main steps involved in the determination of the bulk potential $\fbulk(\vQ)$ and its derivatives: the calculation of the single particle partition function \eqref{eqn:fbulk_density}, and the solution of the mean field self consistency equation \eqref{eqn:multiplier_equation}.  We establish here some of the properties necessary for their calculation.

\subsection{Differentiability}

We require the gradient and Hessian of $\fbulk(\vQ)$ in solving \eqref{eqn:LdG_L2_grad_flow_FE_approx} via Newton's method.

\begin{prop}\label{prop:gradient_hessian_fbulk}
Given $\vQ \in \symmtraceless$ with eigenvalues that satisfy $-1/3 < \lambda_{i} (\vQ) < 2/3$, let $\vA \in \symmtraceless$ be the unique minimizer of \eqref{eqn:dual_convex_func}.  Then, there holds
\begin{equation}\label{eqn:gradient_hessian_fbulk}
\begin{split}
\fbulk (\vQ) &= \vQ \dd \vA - \ln \partition{\vA} = - \dual{\vA}, \\
\frac{\partial \fbulk (\vQ)}{\partial \vQ} \dd \vP &= \vA \dd \vP, \quad \text{for all } \vP \in \symmtraceless, \\
\vT \dd \frac{\partial^2 \fbulk (\vQ)}{\partial \vQ^2} \dd \vP &= \left( \frac{\partial \vA}{\partial \vQ} \dd \vT \right) \dd \vP, \quad \text{for all } \vP, \vT \in \symmtraceless,
\end{split}
\end{equation}
where $\partial \vA / \partial \vQ \dd \vT$ is the unique solution of the linear system
\begin{equation}\label{eqn:hessian_fbulk_linear_sys}
\begin{split}
\left( \frac{\partial \vA}{\partial \vQ} \dd \vT \right) \dd \frac{\partial^2 \dual{\vA}}{\partial \vA^2} \dd \vP = \vT \dd \vP, \quad \text{for all } \vP \in \symmtraceless,
\end{split}
\end{equation}
for any $\vT \in \symmtraceless$, where $\partial^2 \dual{\vA} / \partial \vA^2$ is the constant 4-tensor evaluated at $\vA$.
\end{prop}
\begin{proof}
Let $\rho$ be the probability distribution given by \eqref{eqn:fbulk_density}.  Then, by \eqref{eqn:fbulk_value},
\begin{equation}\label{eqn:gradient_hessian_fbulk_pf_1}
\begin{split}
\fbulk (\vQ) &= \iSp \rho(\vp) \ln \rho (\vp) \dA{\vp} = \iSp \frac{\exp \left( \vp\tp \vA \vp \right)}{\partition{\vA}} \left(  \vp\tp \vA \vp - \ln \partition{\vA} \right) \dA{\vp}, \\
&= \vA \dd \iSp \left( \vp \otimes \vp - \frac{1}{3} \vI \right) \rho(\vp) \dA{\vp} - \ln \partition{\vA} = \vQ \dd \vA - \ln \partition{\vA} = - \dual{\vA},
\end{split}
\end{equation}
where we used \eqref{eqn:defn_of_Q} and \eqref{eqn:dual_convex_func}.  Next, using \eqref{eqn:gradient_hessian_fbulk_pf_1}, 
\begin{equation}\label{eqn:gradient_hessian_fbulk_pf_2}
\begin{split}
\frac{\partial \fbulk (\vQ)}{\partial \vQ} \dd \vP &= \vP \dd \vA + \vQ \dd \vA' - \frac{1}{\partition{\vA}} \frac{\partial \partition{\vA}}{\partial \vA} \dd \vA' = \vP \dd \vA + \vQ \dd \vA' - \left( \vQ + \frac{1}{3} \vI \right) \dd \vA' \\
&= \vP \dd \vA + \vQ \dd \vA' - \vQ \dd \vA' = \vA \dd \vP,
\end{split}
\end{equation}
where $\symmtraceless \ni \vA' \equiv (\partial \vA / \partial \vQ) \dd \vP$ and we used \eqref{eqn:implicit_eqn_lagr_mult_pf_2}.

By differentiating \eqref{eqn:gradient_hessian_fbulk_pf_2}, we clearly get the last line of \eqref{eqn:gradient_hessian_fbulk}.  Thus, we need a characterization of $\partial \vA / \partial \vQ$.  Recall that $\vA$ is the unique minimizer of \eqref{eqn:dual_convex_func}, i.e. $\vA$ satisfies \eqref{eqn:dual_convex_func_pf_1}.  So, we differentiate \eqref{eqn:dual_convex_func_pf_1} with respect to $\vQ$, in the direction $\vT$:
\begin{equation}\label{eqn:gradient_hessian_fbulk_pf_3}
\begin{split}
\left( \frac{\partial \vA}{\partial \vQ} \dd \vT \right) \dd \frac{\partial^2 \dual{\vA}}{\partial \vA^2} \dd \vP = \vT \dd \vP, \quad \text{for all } \vP \in \symmtraceless.
\end{split}
\end{equation}
\end{proof}

\begin{remark}
The 4-tensor $\partial^2 \dual{\vA} / \partial \vA^2$ is positive definite, but the coercivity constant degrades as $|\vA| \to \infty$.
\end{remark}

\subsection{Optimization Procedure}\label{sec:optimization_procedure}

Given $\vQ$, we describe a procedure to obtain the corresponding $\vA = \vA(\vQ)$, as well as its derivative with respect to $\vQ$.

\subsubsection{Solving the Euler-Lagrange Equation}

Recall \eqref{eqn:dual_convex_func} and define the linear form $F_{\vQ}(\vA;\vP)$ to be the first variation of \eqref{eqn:dual_convex_func} with respect to $\vA$:
\begin{equation}\label{eqn:E-L_linear_form}
F_{\vQ}(\vA;\vP) := \frac{1}{\partition{\vA}} \frac{\partial \partition{\vA}}{\partial \vA} \cdot \vP - \vQ \dd \vP.
\end{equation}
Thus, given $\vQ \in \symmtraceless$, we want to find $\vA$ such that $F_{\vQ}(\vA;\vP) = 0$ for all $\vP \in \symmtraceless$.  In other words, we want to find a zero of the non-linear function $F_{\vQ}(\cdot;\vP)$.  Hence, we apply Newton's method. 

For a given $\vA$, define the bilinear form:
\begin{equation}\label{eqn:Newton_bilinear_form}
\begin{split}
m_{\vA} (\vP, \vT) &:= \frac{\partial}{\partial \vA} F_{\vQ}(\vA;\vP) \dd \vT \\
&= \frac{1}{\partition{\vA}} \vP \dd \frac{\partial^2 \partition{\vA}}{\partial \vA^2} \dd \vT - \frac{1}{\partition{\vA}} \left( \frac{\partial \partition{\vA}}{\partial \vA} \dd \vP \right) \frac{1}{\partition{\vA}} \left( \frac{\partial \partition{\vA}}{\partial \vA} \dd \vT \right).
\end{split}
\end{equation}
Then Newton's method is as follows.
\begin{itemize}
\item Initialize. Set $\vA_{0} \in \symmtraceless$ (can take the zero $3 \times 3$ matrix) and set $k = 0$.
\item While not converged, do:
\begin{enumerate}
	\item Solve the following (linear) variational problem.  Find $\delta \vA_{k+1} \in \symmtraceless$ such that
	\begin{equation}\label{eqn:Newton_method_var_problem}
	\begin{split}
	m_{\vA_{k}} (\delta \vA_{k+1}, \vP) = - F_{\vQ}(\vA_{k};\vP), \quad \forall \vP \in \symmtraceless.
	\end{split}
	\end{equation}
	\item Update. Set $\vA_{k+1} := \vA_{k} + \delta \vA_{k+1}$.
	\item If $|\delta \vA_{k+1}|$ is less than some tolerance, then stop.
	\item Else, set $k \leftarrow k + 1$ and return to Step (1).
\end{enumerate}
\end{itemize}

Let $\vA_{*}$ be the solution, i.e. $F_{\vQ}(\vA_{*};\vP) = 0$ for all $\vP \in \symmtraceless$.  Let $\vA_{*}'(\vT) = (\partial \vA / \partial \vQ) \dd \vT$.  We obtain $\vA_{*}'(\vT)$ as the unique solution of the following variational problem.  Find $\vA_{*}'(\vT) \in \symmtraceless$, for every $\vT \in \symmtraceless$, such that
\begin{equation}\label{eqn:deriv_A_var_problem}
\begin{split}
m_{\vA_{*}} (\vA_{*}'(\vT), \vP) = \vT \dd \vP, \quad \forall \vP \in \symmtraceless.
\end{split}
\end{equation}

\subsubsection{Matrix-Vector Form}

Recall the basis $\{ \basis^{k} \}_{k=1}^{5}$ that spans $\symmtraceless$.  We rewrite the Newton method in terms of this basis.

\begin{itemize}
\item Initialize.  Let $\valpha_{0} \in \R^{5}$, with $\valpha_{0} = (\alpha_{0}^{1}, ..., \alpha_{0}^{5})$, such that $\vA_{0} = \sum_{\ell=1}^{5} \alpha_{0}^{\ell} \basis^{\ell}$.  Can simply take $\valpha_{0} = \vzero$.  Set $k = 0$.

\item While not converged, do:
\begin{enumerate}
	\item Compute.  Let $\vb_{k} \in \R^{5}$ such that $b_{k}^{\ell} = - F_{\vQ}(\vA_{k};\basis^{\ell})$ for $\ell = 1, ..., 5$.  Moreover, let $\vH_{k} \in \R^{5 \times 5}$, i.e. $\vH_{k} = [h_{k}^{i j}]_{i,j=1}^{5}$, such that $h_{k}^{i j} = m_{\vA_{k}} (\basis^{i}, \basis^{j})$ for $1 \leq i,j \leq 5$.
	
	\item Solve for $\delta \valpha_{k+1} \in \R^{5}$:
	\begin{equation}\label{eqn:Newton_method_matrix_form}
	\begin{split}
	\vH_{k} (\delta \valpha_{k+1}) = \vb_{k}
	\end{split}
	\end{equation}
	\item Update. Set $\valpha_{k+1} := \valpha_{k} + \delta \valpha_{k+1}$, and define $\vA_{k+1} = \sum_{\ell=1}^{5} \alpha_{k+1}^{\ell} \basis^{\ell}$.
	\item If $|\delta \valpha_{k+1}|$ is less than some tolerance, then stop.
	\item Else, set $k \leftarrow k + 1$ and return to Step (1).
\end{enumerate}
\end{itemize}

Let $\vA_{*}$ be the solution, i.e. $F_{\vQ}(\vA_{*};\vP) = 0$ for all $\vP \in \symmtraceless$.  Let $\vA_{*}'(\basis^{\ell}) = (\partial \vA / \partial \vQ) \dd \basis^{\ell}$, for each $\ell = 1,..., 5$, and let $(\valpha_{*}')^{\ell} \in \R^{5}$ be such that
\begin{equation}\label{eqn:deriv_A_matrix_basis_expansion}
\vA_{*}'(\basis^{\ell}) = \sum_{k=1}^{5} [(\valpha_{*}')^{\ell} \cdot \ve_{k} ] \basis^{k}.
\end{equation}
Next, let $\vH_{*}$ be the Hessian matrix corresponding to $\vA_{*}$.  Then, we obtain $\vA_{*}'(\basis^{\ell})$ by solving for $(\valpha_{*}')^{\ell}$ the equation
\begin{equation}\label{eqn:deriv_A_matrix_form}
\begin{split}
\vH_{*} (\valpha_{*}')^{\ell} = \ve_{\ell}, ~ \text{ for each } \ell = 1, ..., 5.
\end{split}
\end{equation}
Note: this is a one time solve (there is no Newton iteration).

\subsubsection{Computational Issues}

The main difficulty associated with our method is that, contrary to the case of
the Landau-de Gennes energy (or the Oseen-Frank energy in the director
representation of the nematic), the energy density as a function of $\vQ$ is not
known explicitly. As in other self consistent field theories, practical
computation requires a numerical scheme.  In this case, the
approximation involved in the evalulation of the free energy of a configuration
is due to both the Lagrange interpolation operator $\interp$ in
\eqref{eqn:LdG_L2_grad_flow_FE_approx}, and the fact that integrals on the unit
sphere \emph{cannot be computed exactly}. Hence we must address the following
issues.
\begin{itemize}
\item The integrals must be approximated by quadrature, ideally a high order quadrature rule.
\item The strict convexity of the functional in \eqref{eqn:dual_convex_func} degrades as $|\vA|$ becomes large, and this is exactly the situation that arises when the eigenvalues of $\vQ$ approach the physical limits in  \eqref{eqn:Q_eigenvalue_bounds}.  Thus, any inaccuracies in the computations (e.g. the integrals) can adversely affect the convergence of the Newton method.
\item Furthermore, when $\vA$ becomes large, $\exp \left( \vp\tp \vA \vp \right)$ becomes extremely large.  Even though we divide by $\partition{\vA}$, there could be an intermediate overflow result or inaccuracies.
\end{itemize}

Therefore, we introduce the following modifications of the optimization method described earlier.  Before computing any of the above, first do an eigen-decomposition of $\vQ$.  
If the eigenvalues are in the range:
\begin{equation}\label{eqn:Q_eigenvalue_bounds_restricted}
\begin{split}
-\frac{1}{3} + \delta_0 \leq \lambda_{i} (\vQ) \leq \frac{2}{3} - \delta_0, ~\text{ for } i = 1,2,3,
\end{split}
\end{equation}
for some $\delta_0 > 0$, then the plain Newton method above is sufficient (it converges in $\approx 5$ iterations).  Numerical experience indicates that $\delta_0 = 0.05$ to $0.1$ is adequate. If the eigenvalues are outside the range \eqref{eqn:Q_eigenvalue_bounds_restricted}, then one must use a sufficiently accurate quadrature rule to ensure that the system in \eqref{eqn:E-L_linear_form} is accurately approximated.  In addition, a more robust optimization procedure should be used (e.g., the Broyden-Fletcher-Goldfarb-Shanno algorithm with a line search to ensure the objective function decreases) to account for possible numerical sensitivities.  This is not difficult to implement since the problem size is small. However, we have not explored this possibility yet.

Next, as a general concern, the integrals should be computed using a shifting procedure.  For example, consider the computation of
\begin{equation}\label{eqn:exam_integral_calc_pt_1}
\begin{split}
\frac{1}{\partition{\vA}} \frac{\partial \partition{\vA}}{\partial \vA} &= \frac{1}{\iSp \exp \left( \vp\tp \vA \vp \right) \dA{\vp}} \iSp \left( \vp \otimes \vp \right) \exp \left( \vp\tp \vA \vp \right) \dA{\vp},
\end{split}
\end{equation}
for a given $\vA \in \symmtraceless$.  Let $C_{0} = |\vA|$.  Then, \eqref{eqn:exam_integral_calc_pt_1} is equivalent to
\begin{equation}\label{eqn:exam_integral_calc_pt_2}
\begin{split}
\frac{1}{\partition{\vA}} \frac{\partial \partition{\vA}}{\partial \vA} &= \frac{1}{\iSp \exp \left( \vp\tp \vA \vp - C_{0} \right) \dA{\vp}} \iSp \left( \vp \otimes \vp \right) \exp \left( \vp\tp \vA \vp - C_{0} \right) \dA{\vp}.
\end{split}
\end{equation}
The advantage of \eqref{eqn:exam_integral_calc_pt_2} over \eqref{eqn:exam_integral_calc_pt_1} is that, when $|\vA|$ is large, \eqref{eqn:exam_integral_calc_pt_2} will \emph{not} result in an overflow calculation.

Lastly, all the integrals over the unit sphere have been approximated by Lebedev quadrature \cite{Lebedev_DM1999}.  We ascertain the accuracy of the computation in section \ref{sec:accuracy_Newtons_method}. However, we note that they involve uniformly distributed points, so when the probability distribution $\rho$ becomes very localized, the integration may fail. We have not encountered this problem in our numerical results below, but note that it would be possible to use an adaptive quadrature method instead.  Indeed, one could adapt the quadrature rule depending on the performance of the Newton solve, or on how close the eigenvalues are to the physical limits.  

\section{Results}\label{sec:results}

All simulations were implemented using the Matlab/C++ finite element toolbox
FELICITY \cite{Walker_SJSC2018,Walker_FEL2021}.  For all 3-D simulations, we
used the algebraic multi-grid solver (AGMG) \cite{Notay_ETNA2010,Napov_NLAA2011,Napov_SISC2012,Notay_SISC2012} to solve the
linear systems appearing in Newton's method. In 2-D, we simply used the
``backslash'' command in Matlab.  Numerical calculations were performed 
with Matlab version R2017b on a Haswell processor with a base clock of $2.5$ 
Ghz at the Minnesota Supercomputing Institute. Spatially distributed Newton 
iterations were paralleized over $24$ threads (Matlab \texttt{parfor}). Execution timings given below correspond to this configuration.

In our simulations, we chose $\bulkeps = 1$.  We also tested the method with  smaller values of $\bulkeps$, and a finer mesh, and there were no issues.  The number of iterations needed to relax increased roughly proportional to the decrease in $\bulkeps^2$.  But the end result was the same.


\subsection{Accuracy of Newton's Method}
\label{sec:accuracy_Newtons_method}

We first look at the accuracy of Newton's method described in Sec.
\ref{sec:optimization_procedure} to invert the mean field self consistency
relation, and obtain the Lagrange multiplier $\vA(\vQ)$.  To test this, we run
the procedure for various degrees of the Lebedev quadrature. We use in the test
a tensor $\vQ$ with maximum eigenvalue parametrized as $ \max_{i} (\lambda_{i})
=  (2/3) S_{n}$, for a range of $S_{n}$. We have examined the cases $S_{n}=0.1$,
$S_{n}=0.6$, $S_{n} = 0.97$, and $S_{n} = 0.995$, which is the largest value of
$S_{n}$ for which the algorithm converges.  Note that if $S_{n}=-0.5$ or $1$,
the physical limit of the eigenvalues is reached, $\vQ$ is no longer physical,
and the corresponding $\vA$ diverges.  The maximum degree of the Lebedev
quadrature tested is $5810$, and we denote $\vA$ given by quadrature at this
degree by $\vA_{max}$.  Table \ref{table:quadDiffs} summarizes the results in
terms of the maximum component of the difference $\vA_{max} - \vA$ for various
quadrature degrees.  We find that for quadrature degrees below $~500$, the
eigenvalues of $\vQ$ must be relatively small in order to obtain accurate values
of $\vA$. For $\vQ$ with eigenvalues close to the their physical limit, the
Lebedev quadrature degree must be sufficiently high. Depending on the
value of $\kappa / T$, it may be necessary to use larger degrees of quadrature
or more sophisticated methods to find $\vA$, as illustrated in Figure
\ref{fig:1DFreeEnergy} showing the bulk potential, Eq.  
\eqref{eqn:bulk_pot_explicit}, as a function of $S_{n}$.  As $\kappa / T$
increases, and the liquid crystal becomes more ordered, the equilibrium value of
$S_{n}$ increases. 

\begin{table*}
\caption{\label{table:quadDiffs} Maximum component of the difference $\vA_{max} - \vA$, for $\vA$ given by Newton's method for various Lebedev quadrature degrees and $\vQ$ with various eigenvalues parametrized by $S_{n}$. $\vA_{max}$ is given by Newton's method with maximum quadrature degree $5810$.}
\begin{tabular}{c c c c c}
Degree & $S_{n} = 0.1$ & $S_{n} = 0.6$ & $S_{n} = 0.97$ & $S_{n} = 0.995$\\ \hline
$14$ & $0.04$ & $1.4$ & $49.6$ & No Convergence \\
$86$ & $3.2\times10^{-9}$ & $1.4\times10^{-3}$ & $31.2$ & No Convergence \\
$590$ & $1.8\times10^{-15}$ & $6.8\times10^{-14}$ & $7.1\times10^{-3}$ & No Convergence \\
$2030$ & $7.2\times10^{-15}$ & $5.1\times10^{-14}$ & $5.1\times10^{-12}$ & $0.1$ \\
$3470$ & $2.3\times10^{-14}$ & $6.6\times10^{-14}$ & $1.9\times10^{-12}$ & $3.3\times10^{-4}$ \\ \hline
$|\vA_{max}|$ & $0.82$ & $5.1$ & $58.3$ & $347$
\end{tabular}
\end{table*}

\begin{figure}
	\includegraphics[width = \textwidth]{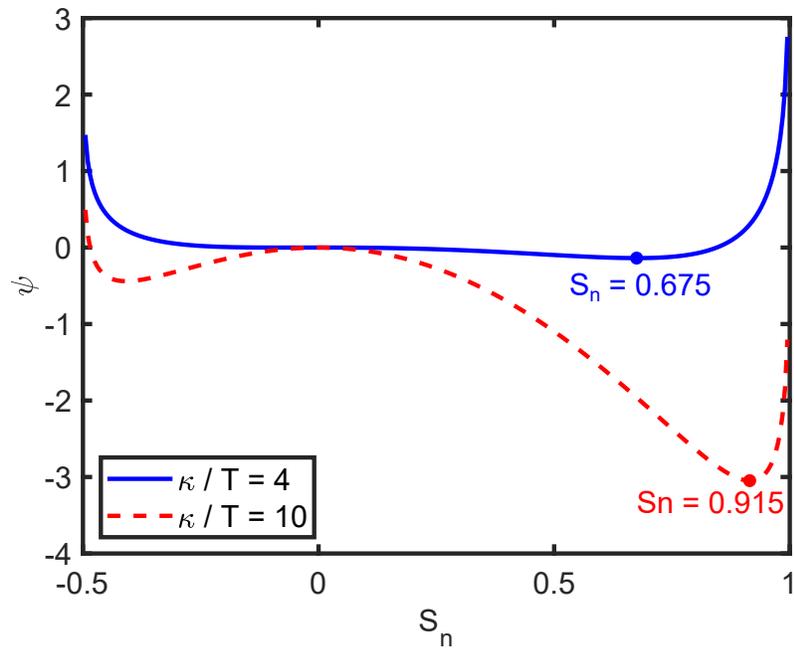}
	\caption{Singular bulk potential, Eq. \eqref{eqn:bulk_pot_explicit}, as a function of $S_{n}$ for $\kappa / T = 4$ and $\kappa / T = 10$.  Dots show the location of the minimum for either case.  As $\kappa / T$ increases, the minimum of the bulk potential approaches the physical limit of $S_{n} = 1$.} 
	\label{fig:1DFreeEnergy}
\end{figure}

\subsection{Boundedness of the bulk potential}

We next examine two spatially nonuniform configurations which are unstable in
the Landau-de Gennes theory when $\Li_{*} \neq 0$, but remain stable when using
the singular bulk potential. In the first example, we choose a weakly perturbed
configuration away from uniform. We take as initial condition a purely uniaxial
configuration defined by $\vQ = S_{n}(x)(\hat{\vn} \otimes \hat{\vn} - 1 / 3
\vI)$ where $\vI$ is the identity matrix, $\hat{\vn} = (0,1,0)$ is fixed, and
$S_{n}(x) = S_0 + \beta \sin{\pi k x}$, where $S_0 = 0.6751$ is chosen so as to
minimize the bulk potential with $\kappa / T = 4$. We set $\beta = 0.1$, $k=10$, $\Li_1 = 1$, $\Li_2 = \Li_3 = \Li_4 = 0$, and $\Li_{*} = 3$. For this ratio of $\kappa/T$, the equilibrium configuration is inside the nematic phase.  Note that the coefficients are just outside the limits stated in Theorem \ref{thm:exist_min}, but the minimizer found appears to be robust.

\begin{figure}
	\includegraphics[width = \textwidth]{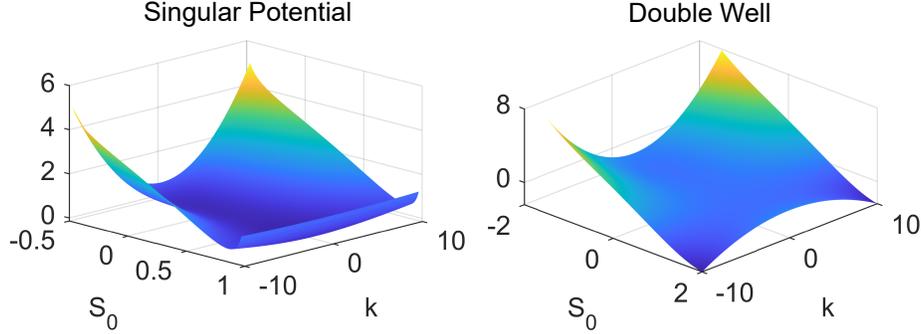}
	\caption{Energy of the perturbed uniform configuration as a function of $S_0$ and $k$ for the Maier-Saupe bulk potential and the double well potential when $L_1 = 1$, $L_2 = L_3 = L_4 = 0$ $L_{*} = 3$.  In the case of the double well, there is a saddle point and unbounded energy, while the Maier-Saupe energy remains bounded below since it diverges if $S \to -0.5$ or $S \to 1$.}
	\label{fig:PertEn}
\end{figure}

We run the gradient flow described in Sec. \ref{sec:gradient_flow} on the components of $\vQ$, decomposed in the basis
\begin{equation} \label{eqn:comp_Basis}
\begin{split}
\vE_1 = 
\begin{pmatrix}
\frac{2}{\sqrt{3}} & 0 & 0 \\
0 & -\frac{1}{\sqrt{3}} & 0 \\
0 & 0 & -\frac{1}{\sqrt{3}}
\end{pmatrix}, \quad
\vE_2 = 
\begin{pmatrix}
0 & 0 & 0 \\
0 & 1 & 0 \\
0 & 0 & -1
\end{pmatrix}, \\
\vE_3 = 
\begin{pmatrix}
0 & 1 & 0 \\
1 & 0 & 0 \\
0 & 0 & 0
\end{pmatrix}, \quad
\vE_4 = 
\begin{pmatrix}
0 & 0 & 1 \\
0 & 0 & 0 \\
1 & 0 & 0
\end{pmatrix}, \quad
\vE_5 = 
\begin{pmatrix}
0 & 0 & 0 \\
0 & 0 & 1 \\
0 & 1 & 0
\end{pmatrix},
\end{split}
\end{equation}
using a square domain defined by $[0,1]^2$ with a body centered mesh with $150 \times 150$ squares ($44701$ vertices), linear basis functions, and a time step in the minimization $\delta t = 4 \times 10^{-3}$. The same parameters are used for both the singular bulk potential and the standard bulk potential of Landau-de Gennes (see \eqref{eqn:Landau-deGennes_bulk_potential}). When $L_{*} \neq 0$, the total LdG energy with standard double well potential is unbounded from below when $k > 0$.  However, the singular bulk potential maintains a bounded free energy for a nonzero range of $L_{*}$ since it diverges outside of the range $-1/2 \leq S_{n} \leq 1$.  This is illustrated in Fig. \ref{fig:PertEn} where we show the total free energy for a set of configurations within a range of $S_0$ and $k$. The (standard) double well energy has a saddle along the line $k=0$ indicating lack of stability for any $k$ and a range of amplitudes $S_{0}$. On the other hand, given the divergence of the Maier-Saupe bulk potential outside the admissible range of eigenvalues of $\vQ$, the free energy computed remains bounded below for all admissible values of $S_{0}$ and $k$.  In fact, the surface plot.

To further illustrate the difference between the two energies, we show in Fig. \ref{fig:PertEvo} the gradient flow of $S_{n}$ during the minimization procedure described in Sec. \ref{sec:gradient_flow}.  The configuration obtained by iterating with the standard Landau-de Gennes double well energy diverges quickly, whereas in the Maier-Saupe case, it simply relaxes to a uniform configuration. 

\begin{figure}
	\includegraphics[width = \textwidth]{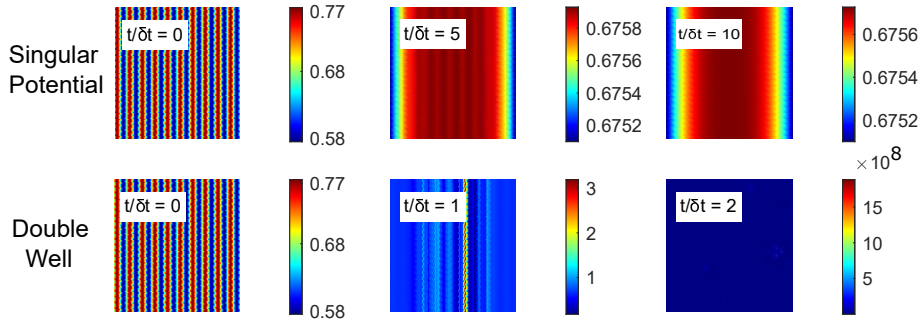}
	\caption{Comparison of the evolution of $S_{n}$ for the perturbed 
uniform configuration given by Maier-Saupe bulk potential, and that given by
the Landau-de Gennes double well potential. We use $k=10$, $\beta = 0.1$, and
$L_{*} = 3$.  The evolution corresponding to the Maier-Saupe bulk potential
relaxes to a uniform configuration, while the configuration evolving under the double well diverges rapidly.}
	\label{fig:PertEvo}
\end{figure}

The second configuration studied is an adaptation of the example from Ball and Majumdar \cite{Ball_MCLC2010} meant to demonstrate the stability of the singular bulk potential.  We consider a cylindrically symmetric initial condition $\vQ = S_{n}(r)(\hat{\vr} \otimes \hat{\vr} - 1 /3 \vI)$ with
\begin{equation}\label{eqn:Ball_Majumdar_exam}
S_{n}(r) = 
\begin{cases}
S_0(2 + \sin{\frac{\pi k r}{5}})            & \quad 0 < r < 5 \\
2 S_0 (2 + \sin{\pi k})(1 - \frac{r}{10})   & \quad 5 < r < 10. 
\end{cases}
\end{equation}
This initial condition is allowed to relax by gradient flow as in Sec. \ref{sec:gradient_flow}.  The value of $S_0 = 0.32$ is chosen so that the eigenvalues of $\vQ$ are close to the physically admissible limit, and $\kappa / T = 3$ so that the bulk potential is minimized for the isotropic phase $S = 0$ \cite{Schimming_PRE2020a}.  We also set $k=5$, $L_1 = 1$, $L_2 = L_3 = L_4 = 0$, and $L_{*} = 3$. A body centered mesh with $150\times150$ squares in a square domain with bounds $[-10,10]^2$, and time step $\delta t = 4\times10^{-3}$ are used.  Each iteration of the gradient flow for this mesh size takes $\sim 30$ CPU minutes to complete.  By direct substitution of Eq.  \eqref{eqn:Ball_Majumdar_exam} into the Landau-de Gennes free energy,  Ball and Majumdar showed that the energy is unbounded below if there is no constraint on the value of $S_0$ when $L_{*} \neq 0$ \cite{Ball_MCLC2010}.  Figure \ref{fig:BallMajumdar} shows several time steps in the in the gradient flow of $S_{n}$ for the initial condition \eqref{eqn:Ball_Majumdar_exam} with both the singular bulk potential and the standard double well Landau-de Gennes potential.  As expected, the flow corresponding to the standard double well potential fails to converge when $L_{*} \neq 0$, whereas the singular bulk potential eventually converges to a configuration with uniform eigenvalues. 

\begin{figure}
	\includegraphics[width = \textwidth]{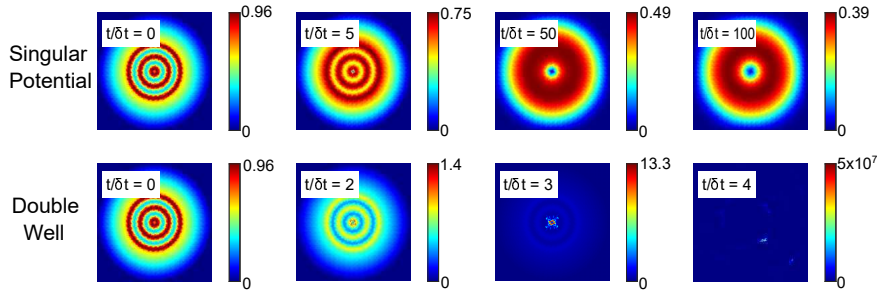}
	\caption{Evolution of $S_{n}$ for the example from Ball and Majumdar \cite{Ball_MCLC2010} for the Maier-Saupe bulk potential and standard bulk potential with $k=5$ and $L_{*} = 3$.  The system with the Maier-Saupe potential eventually relaxes to the isotropic phase, while the double well system diverges rapidly.}
	\label{fig:BallMajumdar}
\end{figure}

\subsection{Three dimensional configurations}

Although the self consistent field theoretic method introduced might appear to lead to a more complex numerical implementation than the Landau-de Gennes theory, we show that even with modest computational resources it is possible to obtain defected configurations in three spatial dimensions. We consider three examples: a $m=+1$ point defect, a line disclination of charge $m=-1/2$, and a Saturn ring loop disclination.  For all calculations we use $\kappa / T = 4$, $L_1 = 1$, $L_2 = L_3 = L_4 = 0$, and $L_{*} = 3$. As above, the equilibrium configuration is in the nematic phase. For the point defect and line disclination, we use a cubic domain with bounds $[-5,5]^3$ with a uniform tetrahedral mesh with $41 \times 41 \times 41$ vertices.  For the Saturn ring we use a cubic domain of size $[-30,30]^3$ with a spherical cavity of radius $7.5$ and a body-centered-cubic (bcc) mesh with $127108$ vertices.  For all computations, we use piecewise linear finite elements and a time step $\delta t = 5\times10^{-2}$.  Iteration is continued until the energy change falls within a tolerance of $10^{-6}$.  For the point defect, Dirichlet boundary conditions on the components of $\vQ$ are used on all sides of the computational domain so as to enforce the topological charge of the defect at the center. For the line disclination, Neumann boundary conditions on the components of $\vQ$ are used on the top and bottom of the computational domain, while Dirichlet conditions are used on all lateral sides to maintain the topological charge of the line at the center of the domain.  For the Saturn ring, Dirichlet boundary conditions fixing a uniform configuration with molecules oriented along the $z$-axis are used on the exterior sides of the domain while Dirichlet conditions are used on the interior sphere to fix a configuration with molecules oriented radially.  

\begin{figure}
	\includegraphics[width = \textwidth]{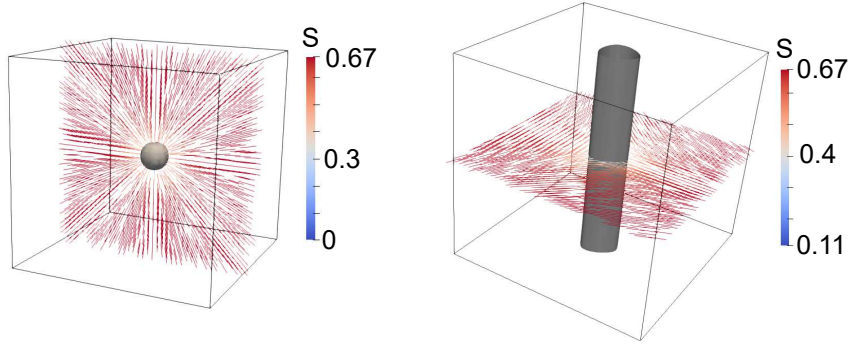}
	\caption{3D visualization of the equilibrium configurations for a $m=+1$ point defect (left) and a $m=-1/2$ line disclination (right).  For both figures, the surface shows all points where $S_{n} = 0.5 S_{max}$ i.e. the ``boundaries'' of the defects.  For both simulations we set $\kappa / T = 4$, $\Li_1 = 1$, $\Li_2 = \Li_3 = \Li_4 = 0$ and $\Li_{*} = 3$.}
	\label{fig:Defects}
\end{figure}

Figure \ref{fig:Defects} shows 3D visualizations of equilibrium configurations for a $m=+1$ point defect and a $m=-1/2$ line disclination. Both simulations reach equilibrium in $\sim12$ CPU hours. The surfaces in both figures show all points where $S_{n} = 0.5 S_{max}$ which we define as the ``boundary'' of the defect.  Note that for the line disclination, the eigenvalue profile is not isotropic due to the inclusion of cubic order terms in the elastic free energy. Also note that the defect core becomes biaxial, that is, the eigenvalues of $\vQ$ become distinct.  These two features are shown clearly in Figure \ref{fig:DiscDist}, which shows a cut in $z=0$ plane of $S_{n}$ along with the molecular orientation probability distribution, $\rho(\vp)$, at various points through the defect.  Far from the defect, the distribution is uniaxial and the corresponding $\vQ$ has two degenerate eigenvalues.  As the core of the defect is approached, the distribution spreads out in the $xy$ plane and the corresponding $\vQ$ has three distinct eigenvalues.  At the center of the defect, the distribution is once again uniaxial but now corresponds to a disk with all orientations in the $xy$ plane equally weighted, which is distinct from the commonly referred to phenomenon of ``escape to the third dimension'' \cite{deGennes_book1995}.  This biaxiality and the anisotropy is consistent with experimental observations in chromonic lyotropic liquid crystals \cite{Zhou_NC2017}, and has been discussed in detail in \cite{Schimming_PRE2020b}.

\begin{figure}
	\includegraphics[width = \textwidth]{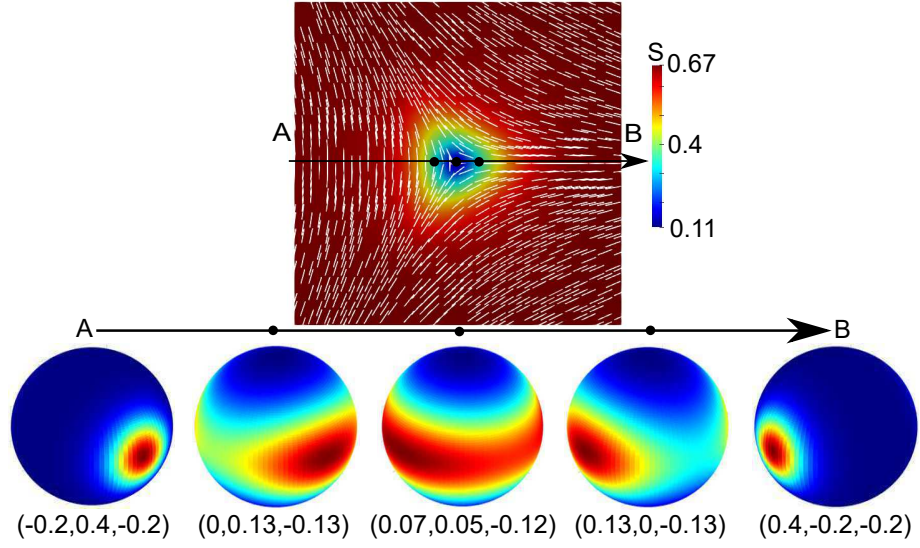}
	\caption{Cut of the equilibrium configuration of $S_{n}$ in the $z=0$ plane for a disclination line with plots of the molecular orientation probability distribution, $\rho(\vp)$, at various points through the disclination.  The distribution becomes biaxial as the core of the defect is approached. At the center of the core all orientations in the $xy$ plane are equally probable. The eigenvalues of $\vQ$ are listed below the corresponding distribution as $(\lambda_{1},\lambda_{2},\lambda_{3})$.}
	\label{fig:DiscDist}
\end{figure}

Finally, Figure \ref{fig:SatRing} shows a 3D visualization and a cut through the $x=0$ plane of a Saturn ring loop disclination around a particle with homeotropic anchoring.  As before, the surface shows all points on the ``boundary'' of the defect while the cut shows the value of $S_{n}$. The shape of the defect and the director field with a characteristic $m=-1/2$ charge are consistent with previous investigations of the Saturn ring \cite{Alama_PRE2016,Gu_PRL2000}.  The Saturn ring configuration takes $\sim84$ CPU hours to reach equilibrium.  This is longer than the other 3D simualations because the mesh is larger and the system requires more iterations to reach equilibrium.

\begin{figure}
	\includegraphics[width = \textwidth]{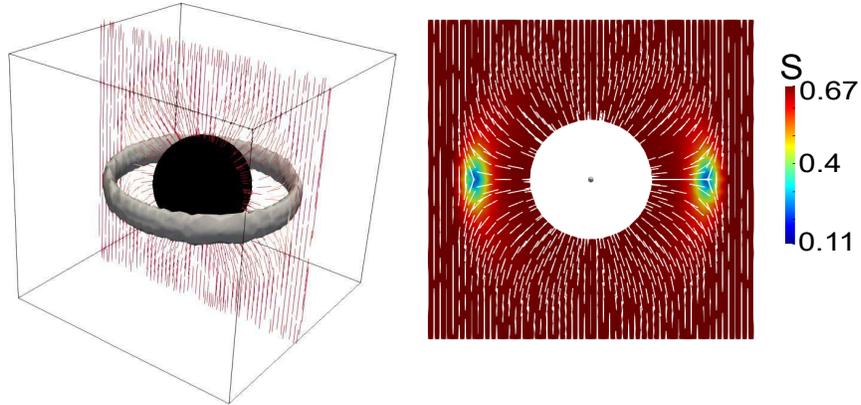}
	\caption{3D visualization of the equilibrium configuration for a Saturn ring loop disclination (left) and plot of $S_{n}$ in the $x=0$ plane for the same configuration (right).  The surface shows all points where $S_{n} = 0.5 S_{max}$ i.e. the ``boundary'' of the defect.  For the simulation we set $\kappa / T = 4$, $L_1 = 1$, $L_2 = L_3 = L_4 = 0$ and $L_{*} =  3$.}
	\label{fig:SatRing}
\end{figure}

\section{Conclusions and discussion} \label{sec:conclusion}

The analysis of a new computational method to obtain equilibrium configurations of a nematic liquid crystal has been presented. The method, based on the Ball-Majumdar singular bulk potential, can overcome known limitations of the Landau-de Gennes theory in the case of elastically anisotropic media. We present selected numerical results demonstrating the convergence of the method in cases in which the Landau-de Gennes theory fails, and a study of prototypical three dimensional configurations that include both point and line topological defects. The code developed has been incorporated into the FELICITY finite element framework in order to facilitate adoption \cite{Walker_SJSC2018,Walker_FEL2021}.

The results shown have been obtained with a particular microscopic interaction model: the Maier-Saupe contact potential. This is a simple case to study since the resulting interaction energy is simply quadratic in $\vQ$. However, the extension to more complex interaction energies is possible assuming one knows their explicit functional representation in terms of the mesoscale $\vQ$. The eigenvalue constraint as introduced is captured in the entropy functional, which only depends on the definition of $\vQ$, and hence is independent of the form of the interaction energy functional. Of course, the limitation of a mean field approximation remains as long as the energy only depends on the local statistical average of the molecular orientation $\vp$. This is not expected to be a serious shortcoming as long as thermal fluctuations are negligible. This is the case in the majority of contemporary studies that focus on systems deep inside the nematic phase.

We have restricted our analysis to finding minimizers of the free energy functional, but the method can be readily extended to studies of nematodynamics (including hydrodynamic interactions). One simply needs to replace the Landau-de Gennes functional by the free energy computed from the singular potential. While there is an additional computational cost involved, it is not severe as demonstrated by our calculations of defect configurations in three dimensions, as long as one takes advantage of parallelism. The method can be therefore applied to the study of the temporal evolution of elastically anisotropic systems, including mass flows. Such a capability should be specially relevant to studies of nematic active matter in which the length of the molecular constituents, and the dependence of their elastic constants on the Debye length when charged, leads to strong elastic anisotropy.


%




\bibliography{MasterBibTeX}

\end{document}